\documentclass[a4paper,11pt]{article}
\usepackage{natbib}

\usepackage{fullpage}
\usepackage{times}

\usepackage{url}
\usepackage{iftex}\ifPDFTeX\usepackage[utf8]{inputenc}\fi
\usepackage{color}
\usepackage{amsmath,amssymb,amsthm}
\usepackage{algorithm,empheq}
\usepackage{algpseudocode}

\usepackage{booktabs,array,tabularx,longtable,placeins}
\usepackage{microtype}
\usepackage{enumitem}

\usepackage{tikz}

\usetikzlibrary{shapes,positioning,fit,backgrounds,calc}

\newtheorem{theorem}{Theorem}
\newtheorem{definition}{Definition}

\newtheorem{lemma}{Lemma}
\newtheorem{example}{Example}

\newtheorem{claim}{Claim}

\pdfoutput=1

\usepackage{libertine}

\usepackage[a4paper, total={170mm,257mm}, left=20mm, top=20mm]{geometry}

\newlength\mylen

\def\cA{\mathcal A}

\def\cI{\mathcal I}

\def\cJ{\mathcal J}

\def\cS{\mathcal S}

\newcommand{\p}{\ensuremath{\mathrm{P}}}
\newcommand{\np}{\ensuremath{\mathrm{NP}}}

\newcommand{\MEW}{\ensuremath{\mathsf{MEW}}}
\newcommand{\OPT}{\ensuremath{\mathsf{OPT}}}

\newcommand{\ps}{\ensuremath{\mathsf{ps}}}

\newcommand{\EW}{\mathrm{EW}}

\newcommand{\1}{\mathbf{1}}

\title{Settling the Computational Complexity of Max-Min Allocation with Ternary Valuations}
\author{
Thi Ngoc Anh Vu\thanks{Faculty of Computer Science, Phenikaa University, Hanoi 12116, Vietnam}
\and 
Trung Thanh Nguyen\thanks{DataOptLab, National Economics University, Hanoi 11616, Vietnam} 
\and 
Khaled Elbassioni\thanks{Khalifa University of Science and Technology, Abu Dhabi 127788, UAE} 
\and 
J\"{o}rg Rothe\thanks{Institut f\"{u}r Informatik, MNF, Heinrich-Heine-Universit\"{a}t D\"{u}sseldorf, 40225, D\"{u}sseldorf, Germany} 
}
\date{\today}

\begin{document}
\maketitle

\begin{abstract}
We study the problem of computing an allocation of indivisible items that maximizes egalitarian welfare, i.e., the utility of the worst-off agent, when agents' item values or marginal values belong to a small set. For additive valuations with values in $\{p,q\}$, where $q>p>0$ and $\gcd(p,q)=1$, we give a polynomial-time algorithm when $p=2$ and prove constant-gap hardness when $p\geq3$, already with exactly three high-valued goods per agent. We also give an $\sqrt{3/2}$-approximation for common positive bi-valued additive valuations. For mixed additive valuations in $\{-p,0,c\}$, where $p\in\{1,2\}$ and $c$ is a positive integer, a reduction to maximum-weight perfect matching resolves the conjectured tractability of $\{-2,0,c\}$-valuations. For submodular valuations with marginals in $\{-2,0,c\}$, where $c$ is odd, we establish an exact unit-gap hardness result and exponential value-query lower bounds, even when all but one agent are additive. Finally, for $\{-1,0,1\}$-submodular valuations,  we prove that no finite multiplicative approximation exists unless $\p=\np$. Together, our results resolve open questions and provide a complete picture of the computational complexity of max-min allocation with ternary valuations.
\end{abstract}

\section{Introduction}\label{sec:intro}\label{sec:introduction}

In the fair division of indivisible items, the goal is to allocate a set of $m$ indivisible resources (also referred to as goods with nonnegative value, or chores with nonpositive value) among $n$ agents with heterogeneous preferences, typically represented by valuation functions~\cite{brandt2016handbook,lan-rot:b-2nd-edition:economics-and-computation-fair-division-of-indivisible-goods}. Dating back to the work of~\cite{Ste48}, fair division has become a central topic at the intersection of economics and computer science~\citep{BCE+16,lan-rot:b-2nd-edition:economics-and-computation-fair-division-of-indivisible-goods}, motivated by a wide range of applications. One of the fundamental welfare objectives studied in fair division is \emph{egalitarian welfare}~\citep{Sen2017}, which evaluates an allocation by the utility of its worst-off agent. Thus, maximizing egalitarian welfare ($\MEW$) seeks an allocation that maximizes the minimum utility across all agents and can be viewed as the natural computational counterpart of the Rawlsian maximin principle. This objective arises naturally in applications where protecting the least well-off participant is a primary concern, such as the allocation of public resources, computing resources, scheduling, and humanitarian aid.

Egalitarian welfare has several features that make it particularly important for our purposes. First, it directly provides a guarantee to the worst-off agent, making it a natural objective in settings involving the allocation of essential resources or burdens~\citep{Moulin2004}. Second, unlike Nash welfare, which is based on the product of individual utilities~\citep{Nash50}, egalitarian welfare remains meaningful for mixed manna and chores, where utilities may be negative and the product of utilities generally lacks a meaningful welfare interpretation. Third, egalitarian welfare is closely connected to other fundamental fairness notions. In particular, the maximin share of an agent is precisely the optimal egalitarian welfare obtained when partitioning the items among identical copies of that agent~\citep{Bud11}. 
Unfortunately, $\MEW$ is computationally hard in general, even for additive valuations over goods~\cite{BD05}. This motivates the search for structured valuation classes under which the problem can be solved exactly or approximated efficiently. One fruitful approach to identifying the sources of computational hardness in fair division is to restrict the \emph{set of values} that an agent may assign to individual items. %Informally, one considers with $A$-valuations with $A$ being the set containing all good values and $|A|$ is small. 
In \emph{binary} instances, each item has value $0$ or $1$ to each agent. Under this restriction $\MEW$ becomes polynomial-time solvable~\citep{Golovin2005MaxminFA}. %and Nash-welfare maximization~\citep{BKV18,HPPS20}; 
Related tractability results extend even to submodular valuations induced by matroid-rank functions~\citep{BEF21,VZ23}. \emph{Bi-valued} instances, in which item values belong to $\{a,b\}$, and \emph{ternary} (or \emph{tri-valued}) instances, in which they belong to $\{a,b,c\}$, constitute natural extensions of the binary setting.

On the practical side, such restricted value domains are appealing for several reasons. First, they capture coarse preference elicitation, in which agents report only a small number of preference levels—for example, whether they like, are indifferent to, or dislike an item, as naturally represented by $\{-1,0,1\}$-valuations and their weighted variants. Second, despite their seemingly simple structure, small value sets can already capture much of the computational difficulty of unrestricted instances, while particular choices of values may still admit efficient exact algorithms. They therefore provide a useful framework for isolating the arithmetic and combinatorial features responsible for computational hardness. %Third, restricted value domains have proved fruitful in uncovering structural properties of fair allocations. For example, optimal Nash-welfare allocations satisfy EFX for bi-valued goods~\citep{ABF+21}; EFX and Pareto-optimal allocations can be computed for bi-valued goods~\citep{GM23}; and EF1 and Pareto-optimal allocations exist and can be computed for bi-valued chores~\citep{EPS22,GMQ22}.

A systematic classification for small value sets can pinpoint the boundary between polynomial-time solvability and computational hardness and thereby provide a finer-grained picture of the computational landscape of fair division. Motivated by this question, in this work we study the complexity of $\MEW$ for bi-valued and ternary instances, with both additive and submodular valuations. 

%These results suggest a natural complexity question: \emph{how far does the tractability of egalitarian-welfare maximization extend beyond binary valuations?} A systematic classification for small value sets can pinpoint the boundary between polynomial-time solvability and computational hardness and thereby provide a finer-grained picture of the computational landscape of fair division. Motivated by this question, in this work we study the complexity of maximizing egalitarian welfare for bi-valued and ternary instances, with both additive and submodular valuations.

\subsection{Related work}\label{sec:related}

For general additive valuations over goods, $\MEW$ can be approximated within a factor of $O(m^{\varepsilon})$~\citep{CCK09} for every fixed $\varepsilon>0$, but is $\np$-hard to approximate within any factor smaller than $2$, even when item values in $\{0,1,2\}$~\citep{BD05,Golovin2005MaxminFA}. \cite{CCK09} gave an $O(m^{\varepsilon})$-approximation in polynomial time, where $m$ is the number of items. This gap has remained open for decades. For restricted instances (a.k.a Santa Claus problem), where the value of any item $j$ is either $0$ or $p_j$, constant-factor polynomial-time algorithms are known~\citep{AKS17,ChengMao2018,DaviesRothvossZhang2020}. %These monotone settings differ from mixed manna, where negative marginal values are allowed.
%A number of works have restricted attention to settings where agents use a small number of values to represent their preferences. 
For more restricted case with binary and, more generally, matroid-rank valuations, $\MEW$ can be solved in polynomial time~\citep{Golovin2005MaxminFA,HPPS20,BEF21,VZ23}. For bi-valued goods with values $\{1,c\}$, where $c$ is a positive integer, this remains true for submodular valuations~\citep{CVZ23a}. In tri-valued setting, \cite{Golovin2005MaxminFA}  studied the special case called ``Big Goods/Small Goods'', where the small goods have value $0$ or $1$ to an agent and the big goods have value $0$ or $x>1$, and gave a $O(\sqrt{n})$-approximation. \cite{chan.2016} considered a different two-value restriction: each item has a common heavy or light value, but an agent may value it at zero. They presented a quasi-polynomial-time $(3+4\varepsilon)$-approximation when the two nonzero item values are $1$ and $\varepsilon$. \cite{KhotPonnuswami2007} generalized this setting, in which good values are in $\{0,1,c\}$ for any $c>1$, and obtained an $\frac{n}{\alpha}$-approximation algorithm with $m^{O(1)}n^{O(\alpha)}$ running time, for any $\alpha \leq \frac{n}{2}$. \cite{FVZ25} established constant-gap hardness for every ternary set of distinct nonnegative values. For additive chores, $\MEW$ corresponds to makespan minimization on unrelated machines; the bi-valued case is polynomial when one cost is twice the other, with hardness outside this ratio~\citep{LST90,CKL15}.

For mixtures of goods and chores, \cite{CVZ23b} gave a polynomial-time algorithm for computing \emph{leximin} allocations, and hence for solving $\MEW$, under $\{-1,0,c\}$-valuations with a common positive integer $c$. Their algorithm further extends to \emph{order-neutral} submodular valuations. They also established hardness for additive $\{-p,c\}$-valuations when $p\geq3$ and $c$ are coprime. Subsequently, \cite{FVZ25} investigated the remaining ternary cases, showing that every $\{-1,0,c\}$-submodular valuation with $c\geq2$ is order-neutral and establishing hardness for $\{-1,0,1\}$-submodular valuations. In particular, they conjectured that $\MEW$ is polynomial-time solvable for additive $\{-2,0,c\}$-valuations and left open the complexity of the problem for submodular $\{-2,0,c\}$-valuations.

For other fair division problems with few utility values, we refer to~\citep{Bliem16,BKV18,Halpern20,ACH+22,CVZ23a,Bismuth24,ACH+25,FVZ25,Jiarong26,GM23,ABF+21,AAB+23,AzizBiroLang2019} and the references therein for the state-of-the-art results. %In particular, the integer and half-integer ratios are tractable for Nash welfare under common bi-valued additive valuations~\citep{ACH+25}. The algorithms in this paper concern egalitarian welfare; the mixed additive algorithm also allows agent-dependent values.

\subsection{Our contributions}\label{sec:contrib}

Our results resolve open questions and provide a complete picture of the computational complexity of max-min allocation with ternary valuations. Table~\ref{tab:complexity} summarizes the resulting compltete picture.

\begin{enumerate}
    \item For additive valuations with values in $\{2,q\}$, with $q\ge 3$, we give a polynomial-time algorithm, complementing the known algorithm for $p=1$. For each fixed $p\geq3$, we prove an explicit constant approximation gap, even when every agent values exactly three goods highly, and present a $\sqrt{3/2}$-approximation algorithm.
    \item For mixed additive valuations in $\{-p,0,c\}$, with $p\in\{1,2\}$ and $c\geq1$, we show a polynomial time  reduction to maximum-weight perfect matching, yielding a polynomial-time algorithm. %In particular, this settles the conjectured tractability of the common domain $\{-2,0,c\}$.
    \item For every fixed odd $c\geq1$, we prove a strongly $\np$-hardness result for $\{-2,0,c\}$-submodular valuations, even when all but one agent are additive. %The reduction has optimum exactly $t$ or $t-1$, for a positive polynomially bounded threshold $t$, excluding an FPTAS and additive error strictly below one. In the value-oracle model, exact optimization requires $2^{\Omega(n)}=2^{\Omega(\sqrt m)}$ queries on the constructed family, also for bounded-error randomized algorithms.
    \item For $\{-1,0,1\}$-submodular valuations, we show the nonexistence of approximation algorithm with any finite approximation factor, unless $\p=\np$.
\end{enumerate}

Among others, our results (2) and (3) answer affirmatively the open question posed in~\cite{CVZ23b}, while result (4) strengthens the $\np$-hardness result of the same paper by establishing a stronger inapproximability bound.

\begin{table}[tb]
\centering
\small
\begin{tabularx}{\columnwidth}{
    @{}
    p{0.3\columnwidth}
    >{\raggedright\arraybackslash}X
    >{\raggedright\arraybackslash}X
    @{}
}
\toprule
Valuations & Additive & Submodular\\
\midrule
$\{0,c\}$, $c\ge 1$
    & P~\citep{Golovin2005MaxminFA}
    & P~\citep{BEF21,VZ23}\\

$\{1,c\}$, $c\ge1$
    & P~\citep{CVZ23a}
    & P~\citep{CVZ23a}\\

$\{2,c\}$, odd $c\ge3$
    & {\bf P (Thm.~\ref{thm:main_q3})}
    & Open\\

$\{p,q\}$, $0<p<q$
    & {\bf $\sqrt{\frac32}$-appr.~(Thm.~\ref{thm:rounding})}
    & {\bf No PTAS (Thm.~\ref{thm:hard})}\\

$\{0,1,c\}$, $c\ge 2$
    & No $(2+\epsilon)$-appr. \cite{chan.2016}
    & No $(2+\epsilon)$-appr. \\

$\{p,q,r\}$, 
    & $\np$-hard \cite{CVZ23b}
    & $\np$-hard \\

$\{-1,0,c\}$, $c\ge2$
    & P~\citep{CVZ23b}
    & P~\citep{CVZ23b,FVZ25}\\

$\{-1,0,1\}$
    & P~\citep{CVZ23b}
    & {\bf No appr. (Thm.~\ref{thm:conflict-hardness})}\\

$\{-2,0,c\}$, even $c\ge4$
    & P \citep{CVZ23b}
    & P ~\citep{CVZ23b,FVZ25}\\

$\{-2,0,2\}$
    & P ~\citep{CVZ23b}
    & {\bf No approx. (Thm.~\ref{thm:conflict-hardness})}\\

$\{-2,0,c\}$, odd $c$
    & {\bf P (Thm.~\ref{thm:mixed-exact})}
    & {\bf No FPTAS (Thm.~\ref{thm:submod-odd})}\\

$\{-p,0,c\}$
    & NP-hard~\citep{CVZ23b}
    & NP-hard~\citep{CVZ23b}\\
\bottomrule
\end{tabularx}
\caption{Complexity of $\MEW$ with ternary valuations. Our new results are highlighted in bold.}
\label{tab:complexity}
\end{table}

\section{Preliminaries}\label{sec:prelim}\label{sec:preliminaries}
Let $N$ be a set of $n$ agents and $M$ a set of $m$ indivisible items. Each agent $i\in N$ is endowed with a valuation function $v_i:2^M\rightarrow \mathbb{Q}$, which assigns a rational value to every subset of items. We assume that $v_i(\varnothing)=0$ for every $i\in N$. An allocation $A=(A_1,\ldots,A_n)$ is a partition of $M$ into $n$ pairwise disjoint bundles, where $A_i$ denotes the bundle allocated to agent $i$. The egalitarian welfare ($\EW$) of an allocation $A$ is defined as 
\[
\EW(A)=\min_{i\in N}v_i(A_i).
\]
We study the problem of finding an allocation that maximizes the egalitarian welfare, which we refer to as the Maximum Egalitarian Welfare ($\MEW$) problem. Denote by $\OPT$ the optimal value of $\MEW$. 
For $t\in\mathbb{Z}_{\ge 0}$, we write $[t]$ for $\{1,\ldots,t\}$. %We introduce  classes of valuation functions that will be studied throughout the paper.

\begin{definition}[$\cA$-submodular]
Given a set of integers $\cA\subseteq\mathbb{Z}$, a valuation function $v$ is
\emph{$\cA$-submodular} if:
\begin{enumerate}
    \item[(a)] for any $o\in M$ and
    $S\subseteq T\subseteq M\setminus\{o\}$, $v(S\cup\{o\})-v(S)
        \geq v(T\cup\{o\})-v(T)$;
    \item[(b)] for any $o\in M$ and $S\subseteq M\setminus\{o\}$, $v(S\cup\{o\})-v(S)\in \cA$
\end{enumerate}
In simple words, the valuation function is submodular, and the marginal
gains are restricted to values in $\cA$. 
\end{definition}

A special case is when $v(S)=\sum_{o\in S} v(o)$ for every $S\subseteq M$, which we call $\cA$-additive valuation. We also consider a class of special $\cA$-submodular valuations, called {\em conflict-pair valuations}, which is formally defined as follows.
\begin{definition}[Conflict-pair valuations]\label{def:conflict-pair}
    A valuation $v_i$ is a \emph{conflict-pair valuation} if $v_i(o) \in \{-1,0,1\}$ for all items $o$ and there is a set $\Pi_i$ of pairwise disjoint pairs $(x,y)$ with $v_i(x)=1$ and $v_i(y)=0$, such that
    \[
    v_i(S) = \sum_{o\in S}v_i(o) \;-\; \bigl|\{(x,y)\in\Pi_i : x,y \in S\}\bigr| ,
    \]
    for every non-empty subset $S\subseteq M$.
\end{definition}
 
In words, $v_i$ is additive except that holding both items of a conflict pair costs one unit: the good $x$ and the item $y$, which is harmless on its own, cannot both be enjoyed. Since the pairs are disjoint, $v_i$ is a sum of functions on disjoint sets of items, each of which is either a single item or the gadget on one pair; hence $v_i$ is submodular and all its marginal values lie in $\{-1,0,1\}$. A formal proof and an example of the valuation can be found in Appendix.

\begin{lemma}\label{lem:conflict-pair-submod}
    Every conflict-pair valuation $v_i$ is $\{-1,0,1\}$-submodular.
\end{lemma}

\begin{definition}[Order-neutrality]\label{def:order-neutral}
    A submodular valuation is order-neutral if, for every bundle $S$, the multiset of marginal gains obtained by adding the items of $S$ is the same for every ordering of those items.
\end{definition}

For submodular valuations, we assume the standard value-oracle model, in which each query specifies an agent and a bundle of items and returns the agent's value for that bundle. %A bounded-error randomized algorithm is required to return a correct answer with probability at least $2/3$ on every instance. 
In Section~\ref{sec:submod}, we use matroid rank functions and their standard properties; we refer the reader to \citet{Oxl11} for further background.

\begin{definition}[Approximation algorithm]\label{def:approx_alg}
    For a maximization problem with nonnegative optimum, a factor-$\rho$ approximation, $\rho\ge1$, returns an allocation with welfare at least $\OPT/\rho$. A PTAS achieves a factor $1+\varepsilon$ in polynomial time for each fixed $\varepsilon>0$; an FPTAS runs in time polynomial also in $1/\varepsilon$. Our multiplicative impossibility statements for mixed manna are proved on explicit nonnegative- or positive-optimum promise classes. An additive error $\delta\ge0$ means returning an allocation $A$ with $\EW(A)\ge\OPT-\delta$. A constant approximation gap rules out a PTAS unless $\p=\np$.
\end{definition}

Consider the decision version of $\MEW$: given an integer threshold $T$, the goal is to either find an allocation whose egalitarian welfare is at least $T$, or correctly determine that no such allocation exists. If this decision problem can be solved in polynomial time, then an optimal allocation for $\MEW$ can be computed in polynomial time by performing a binary search over the range of possible values of $T$. Consequently, throughout the following sections, to establish the polynomial-time solvability of $\MEW$, it suffices to provide a polynomial-time algorithm for the decision version.

%==================================================================================
\section{Bi-valued valuations}\label{sec:bivalued}
%==================================================================================

In this section we assume that each agent $i$ has a $\{p,q\}$-additive valuation with $q>p>0$. 

\subsection{A polynomial-time algorithm for $p=2$}

We consider the case when $p=2$ and $q\ge3$ is odd. Let $q=2s+1,\ s\ge1 $. For agent $i$, let $H_i=\{g\in M: v_i(g)=q\}$ be the set of goods for which $i$ values \emph{high}; all other goods are \emph{low} for $i$.

\begin{theorem}\label{thm:main_q3}
For every odd $q\ge3$, $\MEW$ under $\{2,q\}$-additive valuations can be solved exactly in time polynomial in $n$ and $m$.
\end{theorem}
 
To prove the theorem, it suffices to fix an integer target $T\geq 0$ and show that one can determine in polynomial time whether there exists an allocation with value at least $T$ and, if so, construct such an allocation. Our main idea is to reduce this problem to computing a maximum-weight matching in an appropriately constructed graph; the details of the reduction are presented in the following sections. Since the case $T=0$ is trivial, we henceforth assume that $T\geq 1$. Let $K=\lceil T/q \rceil$, which is the minimum number of high-valued goods that, by themselves, suffice for any agent to attain value at least $T$.

\subsubsection{Feasibility via cost function}
%======================================================================
 
\begin{definition}
    A \emph{high assignment} is a set $X$ of pairs $(g,i)$ with $g\in H_i$ in which every good appears at most once. We write $k_i(X)$, or simply $k_i$, for the number of pairs in $X$ involving agent $i$.
\end{definition}
 
Define the \emph{cost function} 
\[
\Phi(k)=k+\left\lceil\frac{1}{2}(T-qk)^+\right\rceil,~(k\in\mathbb{N}),
\]
where $a^+=\max(a,0)$. It is the minimum number of goods an agent needs to reach utility $T$ when exactly $k$ of its goods are high: the $k$ high goods give a value of $qk$, and each further good gives a value of at least $2$. One can check that  $\Phi(k)\ge \Phi(K)=K$ for all $k\ge K$.
 
\begin{lemma}[Characterization]\label{lem:char}
    $\OPT\ge T$ if and only if there is a high assignment $X$ with $\sum_{i\in N}\Phi(k_i)\le m$. Given such an $X$, an allocation $A$ with $\EW(A)\ge T$ can be constructed in linear time.
\end{lemma}
 
\begin{proof}
    ($\Leftarrow$) Give each agent $i$ the $k_i$ (high) goods assigned to it by $X$, together with $\lceil (T-qk_i)^+/2 \rceil\rceil$ further goods not used by $X$; the inequality ensures enough such goods exist. The remaining goods (if exist) can be distributed arbitrarily. Every agent $i$ values each further good at least $2$, so it receives a value of at least $qk_i+(T-qk_i)^+\ge T$.
 
    ($\Rightarrow$) Suppose that there is an optimal allocation $A$ with $\OPT=\EW(A)\ge T$. Define the set $X:=\{(g,i): g\in A_i\cap H_i,i\in N\}$, a high assignment with $k_i=|A_i\cap H_i|$. Agent $i$ values the other $|A_i|-k_i$ goods of $A_i$ at exactly $2$, so $qk_i+2(|A_i|-k_i)\ge T$. As $|A_i|-k_i$ is a nonnegative integer, $|A_i|-k_i\ge\lceil(T-qk_i)^+/2\rceil$, i.e.\ $|A_i|\ge \Phi(k_i)$. Therefore, $\sum_{i\in N} \Phi(k_i)\le\sum_{i\in N}|A_i|=m$.

    Finally, given $X$, one can construct an allocation $A$ with $\EW(A)\ge T$ in linear time.~
\end{proof}
 
Based on Lemma~\ref{lem:char} our problem of computing an allocation of target $T$ is reduced to checking if
\begin{equation}\label{eq:mincost}
  C\;=\;\min\ \left\{{\sum}_{i\in N}\Phi\bigl(k_i(X)\bigr)~|~X\text{ is a high assignment}\right\}\le m.
\end{equation}
Intuitively, $C$ represents the minimum total number of goods required to ensure that every agent attains utility at least $T$. In what follows, we show that $C$ can be computed in polynomial time by reducing its computation to a maximum-weight matching problem in an appropriately constructed weighted graph. This, together with the preceding arguments, completes the proof of Theorem~\ref{thm:main_q3}.

\subsubsection{Saving and the slot decomposition}

For $0\le k\le K$ and for $j\in[K]$, define 
\begin{itemize}
    \item the \emph{saving} $\sigma(k):=\Phi(0)-\Phi(k)$, which is the number of goods an agent saves by having $k$ of its goods be high-valued rather than none. 
    \item the \emph{marginal saving} of the $j$-th high-valued good $\mu_j:=\Phi(j-1)-\Phi(j)$, i.e., how many fewer goods the agent needs to reach $T$ when one more of its goods is high-valued. 
\end{itemize}
By telescoping, $\sigma(k)=\sum_{j=1}^{k}\mu_j$.

We now represent the saving function $\sigma$ by a collection of \emph{slots}, which will become the vertices of a matching instance in Section~\ref{sec:matching_max_weight}.

\begin{definition}[Slots]
A \emph{\bf single slot} has a weight $w\ge 0$ and holds at most one good. A \emph{\bf pair slot} holds at most two goods; its value is $0$, $s-1$ or $2s-1$ when it holds $0$, $1$ or $2$ goods, respectively. A \emph{\bf filling} of a collection of slots with $k$ goods places the $k$ goods into the slots within these capacities. The \emph{\bf value} of a filling is the total weight of the filled single slots plus the values of the pair slots.
\end{definition}
 
\paragraph{The decomposition $\mathcal S$.} Parse $\mu_1,\dots,\mu_K\in\{s-1;s\}$ from left to right. If the next two unparsed terms are $(s-1,s)$, turn them into a pair slot; otherwise turn the next term $\mu_j$ into a single slot of weight $\mu_j$. The resulting ordered list of slots is $\mathcal S$; it has total capacity $K$ and does not depend on the agent. Let $P$ be its number of pair slots. Lemma~\ref{lem:form} below gives the form of the decomposition, which can be easily to check. %A formal proof and an example of the decomposition can be found the Appendix.  

\begin{lemma}\label{lem:form}
	If $K\ge2$, then $\mathcal S$ consists of, in order:
	\begin{enumerate}[label=(\alph*),nosep]
		\item[(i)] a \emph{lead} single slot of weight $s$ if $T$ is odd, and none if $T$ is even;
		\item[(ii)] $P$ pair slots (possibly empty);
		\item[(iii)] at most two \emph{tail} single slots, each of weight at most $s-1$, in nonincreasing order of weight.
	\end{enumerate}
\end{lemma}

We will use the following result in our analysis of the computation of a maximum-weight matching in the graph.
 
\begin{lemma}\label{lem:slots}
    For every $0\le k\le K$, the maximum value of a filling of $\mathcal S$ with $k$ goods equals $\sigma(k)$.
\end{lemma}

\begin{example}\label{ex:slots}
	Let $q=5$ (so $s=2$) and $T=17$. Then $K=\lceil17/5\rceil=4$ and
	$\bigl(\Phi(0),\Phi(1),\Phi(2),\Phi(3),\Phi(4)\bigr)=(9,7,6,4,4)$, $(\mu_1,\mu_2,\mu_3,\mu_4)=(2,1,2,0)$. 
	Here $T$ is odd, so $\mu_1=s=2$ becomes the lead single slot; the block $(\mu_2,\mu_3)=(s-1,s)=(1,2)$ becomes a pair slot; and $\mu_4=0<\nu_4=1$ is truncated, because three high goods (value $15$) plus one low good (value $2$) already reach $17$, and becomes a tail single slot of weight $0$. Thus $P=1$ and $\mathcal S$ consists of a lead slot, one pair slot and one tail slot (Figure~\ref{fig:decomposit}). The maximum filling values for $k=0,\dots,4$ are $0,2,3,5,5$, which equal $\sigma(k)$: one good goes to the lead slot; two goods give either lead plus half a pair or a full pair, both worth $3$; three goods fill the lead and the pair; the fourth good adds nothing.
\end{example}

\begin{figure}[ht]
	\centering
		\begin{tikzpicture}[
			cell/.style={draw, minimum width=1.9cm, minimum height=0.75cm, inner sep=2pt},
			leadc/.style={cell, fill=violet!12, draw=violet!70!black},
			pairc/.style={cell, fill=teal!12, draw=teal!70!black},
			tailc/.style={cell, fill=orange!12, draw=orange!70!black},
			lab/.style={anchor=east, font=\small},
			>=stealth]
			\foreach \j/\x in {1/0,2/2.2,3/4.4,4/6.6}
			\node[font=\small] at (\x,0.75) {$j=\j$};
			\node[lab] at (-1.2,0) {$\mu_j$};
			\node[leadc] (m1) at (0,0)   {$2$};
			\node[pairc] (m2) at (2.2,0) {$1$};
			\node[pairc] (m3) at (4.4,0) {$2$};
			\node[tailc] (m4) at (6.6,0) {$0$};
			\node[font=\scriptsize, below=0pt of m4] {truncated};
			\node[lab] at (-1.2,-1.8) {slots $\mathcal S$};
			\node[leadc, minimum height=1.0cm, align=center] (L) at (0,-1.8) {lead single\\[-1pt]{\scriptsize weight $s=2$}};
			\node[pairc, minimum width=4.1cm, minimum height=1.0cm, align=center] (P) at (3.3,-1.8)
			{pair slot\\[-1pt]{\scriptsize values $0,\,1,\,3$}};
			\node[tailc, minimum height=1.0cm, align=center] (T1) at (6.6,-1.8) {tail single\\[-1pt]{\scriptsize weight $0$}};
			\draw[->] ($(m1.south)+(0,-0.35)$) -- (L.north);
			\draw[->] ($(m2.south)!0.5!(m3.south)+(0,-0.35)$) -- (P.north);
			\draw[->] ($(m4.south)+(0,-0.35)$) -- (T1.north);
		\end{tikzpicture}
	\caption{The decomposition $\mathcal S$ for $q=5$ ($s=2$) and $T=17$, so $K=4$. Since $T$ is odd, $\mu_1=s$ becomes the lead single slot; the block $(\mu_2,\mu_3)=(s-1,s)$ becomes a pair slot; and the truncated last term $\mu_4=0$ becomes a tail single slot.
	}
	\label{fig:decomposit}
\end{figure}
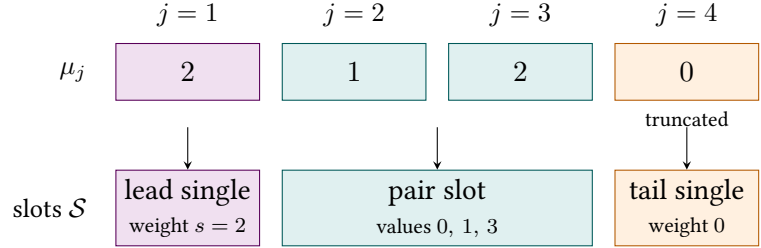

%======================================================================
\subsubsection{Reduction to maximum-weight matching}\label{sec:matching_max_weight}
%======================================================================
 
Build an undirected graph $G$ with edge weights as follows. % (see Example~\ref{ex:slots} and Figure~\ref{fig:gadget} in Appendix for an illustration).

\smallskip
\begin{itemize}
    \item There is a good vertex $g$ for every good $g\in M$.
    \item There is a vertex $u$ for every pair of agent $i$ and single slot of weight $w$ in $\mathcal S$; create an edge $(u,g)$ of weight $w$ for every $g\in H_i$.
    \item There are two vertices $u,u'$ joined by an edge of weight $1$ for every pair of agent $i$ and pair slot in $\mathcal S$; join and each of $u,u'$ to every $g\in H_i$ by an edge of weight $s$. This is done for every $i$.
\end{itemize}

\begin{claim}
    The constructed graph $G$ has at most $m+2nK\le m+2nm$ vertices and $O(nm^2)$ edges.
\end{claim}
 See Figure~\ref{fig:gadget} for an illustration of the graph construction.
  It is observed that goods are adjacent only to slot vertices, and slot vertices of different slots are never adjacent, so every edge of a matching is either an internal edge $(u,u')$ of a pair slot or joins a good to a slot vertex of some agent. 

\begin{figure}[ht]
\centering
\begin{tikzpicture}[
  slotv/.style={circle, draw, minimum size=0.8cm, inner sep=0pt},
  leadv/.style={slotv, fill=violet!12, draw=violet!70!black},
  pairv/.style={slotv, fill=teal!12, draw=teal!70!black},
  tailv/.style={slotv, fill=orange!12, draw=orange!70!black},
  goodv/.style={slotv, fill=gray!10, draw=gray!70!black},
  wL/.style={violet!80!black},
  wS/.style={teal!70!black},
  w0/.style={gray, dashed},
  wlab/.style={fill=white, inner sep=0.6pt, font=\tiny, text=black}]
  % goods (shared)
  \node[goodv] (g1) at (5,3)    {$g_1$};
  \node[goodv] (g2) at (5,1.8)  {$g_2$};
  \node[goodv] (g3) at (5,0.6)  {$g_3$};
  \node[goodv] (g4) at (5,-0.6) {$g_4$};
  % agent 1 (left), H_1 = {g1,g2,g3}
  \node[leadv] (l1)  at (0,3)    {$\ell_1$};
  \node[pairv] (u1)  at (0,1.8)  {$u_1$};
  \node[pairv] (v1)  at (0,0.6)  {$u_1'$};
  \node[tailv] (t1)  at (0,-0.6) {$t_1$};
  % agent 2 (right), H_2 = {g2,g3,g4}
  \node[leadv] (l2)  at (10,3)    {$\ell_2$};
  \node[pairv] (u2)  at (10,1.8)  {$u_2$};
  \node[pairv] (v2)  at (10,0.6)  {$u_2'$};
  \node[tailv] (t2)  at (10,-0.6) {$t_2$};
  \foreach \g/\ps in {g1/0.16,g2/0.3,g3/0.44} {
    \draw[wL] (l1) -- node[wlab, pos=\ps] {$2$} (\g);
    \draw[wS] (u1) -- node[wlab, pos=\ps] {$2$} (\g);
    \draw[wS] (v1) -- node[wlab, pos=\ps] {$2$} (\g);
    \draw[w0] (t1) -- node[wlab, pos=\ps] {$0$} (\g);
  }
  \foreach \g/\ps in {g2/0.44,g3/0.3,g4/0.16} {
    \draw[wL] (l2) -- node[wlab, pos=\ps] {$2$} (\g);
    \draw[wS] (u2) -- node[wlab, pos=\ps] {$2$} (\g);
    \draw[wS] (v2) -- node[wlab, pos=\ps] {$2$} (\g);
    \draw[w0] (t2) -- node[wlab, pos=\ps] {$0$} (\g);
  }
  \draw[teal!70!black, line width=1.3pt] (u1) -- node[left, font=\scriptsize] {$1$} (v1);
  \draw[teal!70!black, line width=1.3pt] (u2) -- node[right, font=\scriptsize] {$1$} (v2);
  \node[draw, dashed, rounded corners, fit=(l1)(t1), inner sep=9pt,
        label={[font=\small]above:agent $1$}] {};
  \node[draw, dashed, rounded corners, fit=(l2)(t2), inner sep=9pt,
        label={[font=\small]above:agent $2$}] {};
  \node[draw, dashed, rounded corners, fit=(g1)(g4), inner sep=9pt,
        label={[font=\small]above:goods}] {};
  \node[font=\scriptsize, left=16pt of l1] {lead};
  \node[font=\scriptsize, left=30pt of $(u1)!0.5!(v1)$] {pair};
  \node[font=\scriptsize, left=16pt of t1] {tail};
  \node[font=\scriptsize, right=16pt of l2] {lead};
  \node[font=\scriptsize, right=30pt of $(u2)!0.5!(v2)$] {pair};
  \node[font=\scriptsize, right=16pt of t2] {tail};
  \begin{scope}[shift={(0.6,-1.95)}, font=\scriptsize]
    \draw[wL] (0,0) -- (0.6,0); \node[right] at (0.6,0) {lead: weight $s$};
    \draw[wS] (3.2,0) -- (3.8,0); \node[right] at (3.8,0) {pair: weight $s$};
    \draw[w0] (6.4,0) -- (7.0,0); \node[right] at (7.0,0) {tail: weight $0$};
  \end{scope}
\end{tikzpicture}
\caption{The graph $G$ for $q=5$, $T=17$ (decomposition of Figure~\ref{fig:decomposit}) on an instance with $n=2$ agents and $m=4$ goods, where $H_1=\{g_1,g_2,g_3\}$ and $H_2=\{g_2,g_3,g_4\}$. Each agent receives its own copy of the slot vertices: a lead vertex $\ell_i$, pair vertices $u_i,u_i'$ joined by an internal edge of weight $1$, and a tail vertex $t_i$; these are joined only to the goods in $H_i$, with weights $s=2$, $s=2$ and $0$, respectively. The good vertices are shared, so a matching assigns each good to at most one agent; $g_2$ and $g_3$ are wanted by both agents. Relative to the baseline $1$ of the internal edge, a pair slot contributes $0$, $s-1=1$ or $2s-1=3$ with $0$, $1$ or $2$ goods.}
\label{fig:gadget}
\end{figure}
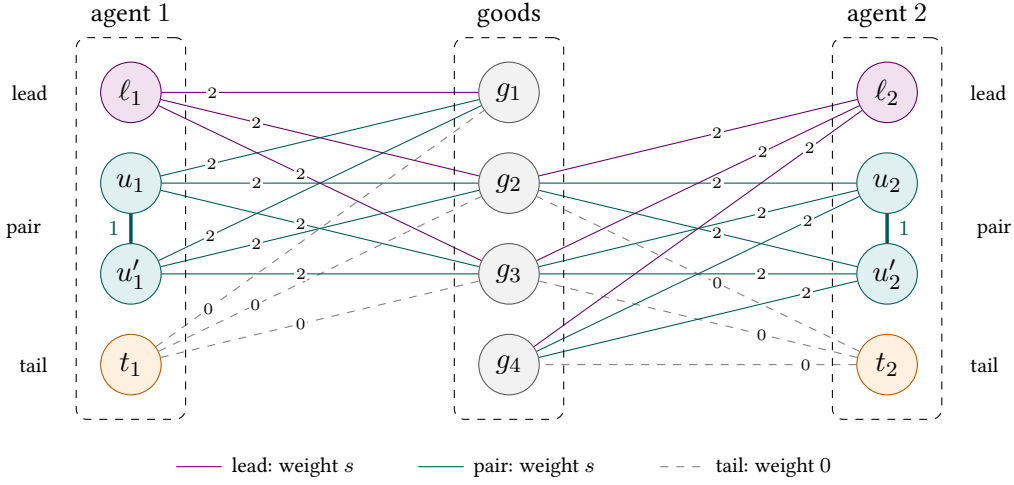

\begin{lemma}\label{lem:matching}
The maximum weight of a matching in $G$ equals
\[
  \omega\;=\;nP+\max_{X\text{ high assignment}}\ {\sum}_{i\in N}\sigma\bigl(\min\{k_i(X),K\}\bigr).
\]
\end{lemma}
 
\begin{proof}
($\le$) Let $Y$ be a matching. For each agent $i$, the goods matched to $i$'s slot vertices lie in $H_i$, are distinct, and are distinct across agents; together they form a high assignment $X$, with $k_i'=k_i(X)\le K$ goods for agent $i$ since the slots have total capacity $K$. The weight of $Y$ splits over agents. A single slot of weight $w$ contributes $w$ if its vertex is matched to a good and $0$ otherwise. A pair slot contributes at most $1$ if neither vertex is matched to a good (only the internal edge can be used), $s=1+(s-1)$ if exactly one is, and $2s=1+(2s-1)$ if both are. Hence agent $i$ contributes at most $P$ plus the value of a filling of $\mathcal S$ with $k_i'$ goods, which by Lemma~\ref{lem:slots} is at most $\sigma(k_i')=\sigma(\min\{k_i',K\})$. Summing over agents gives the inequality.
 
($\ge$) Let $X$ be any high assignment. For each agent $i$, keep $\min(k_i,K)$ of its assigned goods and place them into the vertices of $i$'s slots in parse order; for every pair slot of $i$ that receives no good, add its internal edge. This is a matching, and by Lemma~\ref{lem:form} agent $i$ contributes exactly $P+\sigma(\min\{k_i,K\})$.
\end{proof}
 
\begin{lemma}\label{lem:cost}
    $C=n\,\Phi(0)-(\omega-nP)$, where $C$ is defined in~\eqref{eq:mincost}.
\end{lemma}
 
\begin{proof}
    Removing assigned goods from an agent keeps a high assignment. By definition of $\Phi$, trimming every agent to at most $K$ assigned goods does not increase $\sum_i \Phi(k_i)$, so the minimum in~\eqref{eq:mincost} is attained with all $k_i\le K$. In fact, $\Phi(k_i)=\Phi(\min\{k_i,K\})$ for all $i$. Taking $\Phi(k_i)=\Phi(0)-\sigma(k_i)$ into account, we have that
    \begin{align*}
        C=n\,\Phi(0)-\max_{X\text{ high assignment}}{\sum}_{i\in N}\sigma\bigl(\min\{k_i(X),K\}\bigr),
    \end{align*}
    which is equal to $n\Phi(0)-(\omega-nP)$ by Lemma~\ref{lem:matching}.~
\end{proof}

%\begin{proof}[Proof of Theorem~\ref{thm:main_q3}]
%Fix an integer target $T\ge1$. Construct $\mathcal S_T$ and $G_T$, compute $\omega_T$, and hence $C_T$ by Lemma~\ref{lem:cost}. By Lemma~\ref{lem:char}, $\OPT\ge T$ if and only if $C_T\le m$. In that case, a maximum-weight matching yields, as in the proof of Lemma~\ref{lem:matching}, a high assignment $X$ with $\sum_i \Phi_T(\min(k_i,K))=C_T\le m$ after trimming, and Lemma~\ref{lem:char} turns $X$ into an allocation with $\EW\ge T$.
 
%Regarding the running time of Algorithm~\ref{alg:poly-p2}, 
Note that the graph $G$ has at most $m+2nm$ vertices and $O(nm^2)$ edges. Therefore, a maximum-weight matching in $G$ can be computed in time polynomial in $n$ and $m$~\citep{LovaszPlummer1986}. Consequently, both $\omega$ and $C$ can also be computed in polynomial time.% Moreover, the binary search for the optimal value of $T$ requires only $O(\log m)$ feasibility tests. Combining these observations, we conclude that Algorithm~\ref{alg:poly-p2} finds an optimal allocation for $\MEW$ in polynomial time.
%\end{proof}

%The same conclusion holds for a common positive value set $\{b,a\}$ when $a/b$ is an integer or a half-integer. For a nonintegral half-integer ratio, scale by $2/b$ and apply Theorem~\ref{thm:g-exact}; for an integer ratio, scale by $1/b$ and use the algorithm of \citet{CVZ23a}. Equal values are immediate.

\subsection[Approximability of the case p at least 3]{Approximability of the case $p\ge3$}\label{sec:g-approximation}

In this section, we consider $\{p,q\}$-additive valuations with $q>p>3$. We first show that $\MEW$ is APX-hard in this setting, by using a reduction from $p$-dimensional matching ($p$-DM), which is $\np$-complete for $p\ge 3$~\citep{GareyJohnson1979}, and defined as follows.  %while also establishing that it admits a constant-factor approximation algorithm.
\begin{definition}[$p$-DM]
Given disjoint sets $V_1,\dots,V_p$, each of size $n'$, and a set $E\subseteq V_1\times\dots\times V_p$ of hyperedges, decide whether there are $n'$ pairwise disjoint hyperedges in $E$ (a \emph{perfect matching}).
\end{definition}
 
%For $p>3$, NP-completeness follows by adding $p-3$ copies of $V_3$ and extending every hyperedge by the corresponding copies of its $V_3$-vertex. We may assume $|E|\ge n'$ and that every vertex lies in at least one hyperedge, since otherwise the answer is trivially ``no''.
 
\begin{theorem}[Inapproximability]\label{thm:hard}
For all coprime integers $q>p\ge3$, it is $\np$-hard to approximate the maximum egalitarian welfare with values in $\{p,q\}$ within any factor $\rho<pq/(pq-1)$. %In particular, the problem admits no PTAS unless $\mathrm{P}=\mathrm{NP}$.
\end{theorem}

\begin{proof}
Given a $p$-DM instance $\cI$ with $|E|=m'$ hyperedges, build the following $\MEW$ instance $\cJ$.
\begin{itemize}%[nosep]
  \item \emph{Agents:} one agent $i_e$ for each hyperedge $e\in E$; so $n=m'$.
  \item \emph{Goods:} one \emph{vertex good} $g_v$ for each vertex $v\in V_1\cup\dots\cup V_p$, and $q(m'-n')$ \emph{dummy goods}.
  \item \emph{Values:} $v_{i_e}(g_v)=q$ if $v\in e$; every other value is $p$.
\end{itemize}
This instance $\cJ$ has good values in $\{p,q\}$, and its size is polynomial. 
%\begin{lemma}\label{lem:reduction}
Let $\OPT$ be its optimal value. It suffices to prove that if $\cI$ has a perfect matching then $\OPT=pq$; otherwise $\OPT\le pq-1$.
%\end{lemma}

Suppose that $\cI$ has a perfect matching  $F\subseteq E$. Give each agent $i_e$ with $e\in F$ its $p$ incident vertex goods, worth $pq$ to it. The $pn'$ vertex goods are thereby all allocated. Give each of the remaining $m'-n'$ agents $q$ dummy goods, worth $pq$. This uses all $q(m'-n')$ dummies, and every agent has utility exactly $pq$, so $\OPT\ge pq$. The upper bound $\OPT\le pq$ follows from the averaging argument below.
 
Every vertex good is worth at most $q$ to anyone and every dummy is worth $p$ to everyone. Hence for every allocation $A$,
\begin{equation}\label{eq:sum}
  {\sum}_{i\in N}v_i(A_i)\;\le\;pn'\cdot q+q(m'-n')\cdot p\;=\;pq\,m',
\end{equation}
with equality only if every vertex good is given to an agent that values it $q$. So the average utility is at most $pq$, and $\OPT\le pq$.
 
Suppose that $\cI$ has no perfect matching, and that some allocation $A$ has $\EW(A)\ge pq$. By~\eqref{eq:sum}, equality holds throughout: every agent has utility exactly $pq$, and every vertex good is held by an incident agent. It is seen that if $qk+p\ell=pq$ with integers $0\le k\le p$ and $\ell\ge0$, then $k\in\{0,p\}$. Agent $i_e$ values high only its $p$ incident vertex goods, so $k\le p$ and the claim applies: each agent holds either none or all $p$ of its incident vertex goods. Since all $pn'$ vertex goods are held by incident agents, exactly $n'$ agents hold all $p$ of theirs, and their hyperedges are pairwise disjoint and cover every vertex, i.e.\ they form a perfect matching, a contradiction. Hence $\OPT<pq$, and since utilities are integers, $\OPT\le pq-1$.~
\end{proof}

%\begin{proof}
%Suppose a polynomial-time algorithm always returns an allocation with $\EW\ge\OPT/\rho$ where $\rho<pq/(pq-1)$. On an instance from the construction, if a perfect matching exists it returns $\EW\ge pq/\rho>pq-1$, i.e.\ $\EW\ge pq$ (utilities are integers). If none exists, $\EW\le\OPT\le pq-1$. Comparing the returned welfare with $pq$ decides $p$-DM.
%\end{proof}

%\subsubsection{A $\sqrt{3/2}$-Approximation Algorithm}\label{sec:alg}
%======================================================================
We now to present a constant approximation for $\MEW$.
\begin{theorem}\label{thm:rounding}
    There is a $\sqrt{3/2}$-approximation algorithm for $\MEW$ with $\{p,q\}$-additive valuations, where $3<p<q$.
\end{theorem}
\begin{proof}
Let $r=q/p>1$ and $\cI$ the instance with values in $\{1,r\}$. If $2r\in\mathbb{N}$, the instance can be solved exactly by
Theorem~\ref{thm:main_q3}. Otherwise, let $a=\frac{\lfloor 2r\rfloor}{2},b=\frac{\lceil 2r\rceil}{2},
$ so that $1\le a<r<b$. Let $A^*$ be an optimal allocation for $\cI$, and let $S^a$ and $S^b$ be optimal allocations for the rounded instances $\cI^a$ (which has values in $\{1,a\}$) and $\cI^b$ (which has values in $\{1,b\}$), respectively.

For every agent $i$ and bundle $X$, rounding down gives $\frac{a}{r}v_i(X)\le v_i^a(X)\le v_i(X)$. 
Consequently, $\EW(S^a)\ge \EW^a(S^a)\ge \EW^a(A^*)\ge \frac{a}{r}\OPT$. Similarly, rounding up gives
$\EW(S^b)\ge  \frac{r}{b}\OPT$. Now, returning the allocation $S$ with larger welfare in the
original instance therefore guarantees
\[
\EW(S)\ge
\max\left\{\frac ar, \frac rb\right\}\cdot\OPT \ge \sqrt{\frac ab}\cdot\OPT\ge \sqrt{\frac 23}\cdot\OPT.
\]
Finally, note that both rounded instances have integer or half-integer high
values and are solvable exactly in polynomial time by
Theorem~\ref{thm:main_q3}. Thus $S$ is found efficiently, completing the proof.~
\end{proof}

%==================================================================================
\section[Polynomial Algorithm for Mixed Additive Valuations]{Poly-time Algorithm for $\{-p,0,c\}$-Additivities}
\label{sec:additive}\label{sec:mixed}
%==================================================================================

In this section we will present a polynomial time algorithm for solving $\MEW$, assuming that each agent $i$ has a $\{-p_i,0,c_i\}$-additive valuation, where $p_i\in\{1,2\}$ and $c_i$ is a positive integer. 

\begin{theorem}\label{thm:poly-pc}\label{thm:mixed-exact}
An allocation of maximum egalitarian welfare can be found in time
polynomial in the input length for additive $\{-p_i,0,c_i\}$ valuations,
where $p_i\in\{1,2\}$ and $c_i$ is a positive integer. %The algorithm uses $O(1+\log(n(m+1)))$ maximum-weight perfect-matching computations on graphs with $O(n(m+1))$ vertices and $O(nm^2+n^2)$ edges. The number of matching computations and the graph sizes are independent of the magnitudes of the valuations.
\end{theorem}

Fix an integer threshold $T$. We reduce the decision version of $\MEW$ to the problem of finding a maximum-weight perfect matching in a suitably constructed graph. We next describe this graph and show that the weight of a maximum-weight perfect matching characterizes the feasibility of the threshold $T$.

For every item $o$, let
$G_o=\{i:v_i(o)=c_i\}$ and $Z_o=\{i:v_i(o)=0\}$. 
We classify an item as 
\begin{itemize}
    \item a \emph{good} if $G_o\ne\varnothing$,
    \item a \emph{neutral item} if $G_o=\varnothing$ and $Z_o\ne\varnothing$,
    \item a \emph{universal chore} if $G_o=Z_o=\varnothing$. 
\end{itemize}
 Thus a universal chore
has value $-p_i$ to every agent $i$. Let $H$ be the number of universal
chores. Define a bipartite incidence graph $\Gamma$ between goods and
agents, with $(o,i)$ an edge precisely when $i\in G_o$, and let $d_i$ be
the degree of agent $i$. A \emph{good assignment} is a map $\sigma$ sending
each good $o$ to an agent in $G_o$; it induces the vector
$g(\sigma)=(g_1,\ldots,g_n)$, where $g_i=|\sigma^{-1}(i)|$.

\begin{lemma}[Normalization]\label{lem:normal}\label{lem:mixed-normalization}
Every allocation can be transformed, without decreasing any agent's
value, so that each good goes to an agent in $G_o$ and each neutral item
goes to an agent in $Z_o$.
\end{lemma}
\begin{proof}
If a good $o$ is held by $h\notin G_o$, remove it from $h$ and give it to
any $j\in G_o$. Its removal cannot decrease $h$'s value, and $j$ gains
$c_j$. Similarly, a neutral item held by $h\notin Z_o$ has value $-p_h$
to $h$ and can be moved to any $j\in Z_o$. These moves treat distinct
items independently because the valuations are additive.
\end{proof}

In a normalized allocation, agent $i$ holds $g_i$ goods, each worth
$c_i$, some neutral items worth zero, and $h_i$ universal chores, each
worth $-p_i$. Hence $v_i(A_i)=c_i g_i-p_i h_i.$ 
Define
\begin{equation}\label{eq:mixed-capacity}
L_i=\max\{0,\lceil T/c_i\rceil\},\quad\text{and }~~
f_i(g_i)=\left\lfloor(c_i g_i-T)/p_i\right\rfloor.
\end{equation}
For a nonnegative integer $g_i$, the inequality $c_i g_i\ge T$ is
equivalent to $g_i\ge L_i$. In that case $f_i(g_i)\ge0$ is the
largest number of universal chores that agent $i$ can absorb while
keeping value at least $T$. These statements also hold when $T$ is
negative.

\begin{lemma}[Chore absorption]\label{lem:absorb}\label{lem:mixed-absorption}
    There is an allocation of egalitarian welfare at least $T$ if and only if there is a good assignment $\sigma$ such that, for $g=g(\sigma)$, $g_i\ge L_i$ for all $i$, and $\sum_{i=1}^n f_i(g_i)\ge H$.
\end{lemma}
\begin{proof}
Suppose there is an allocation of welfare at least $T$, and normalize it
using Lemma~\ref{lem:mixed-normalization}. If agent $i$ receives $g_i$
goods and $h_i$ universal chores, then $c_i g_i-p_i h_i\ge T$.
Since $g_i$ is a nonnegative integer and $c_i g_i\ge T$, we have
$g_i\ge\max\{0,\lceil T/c_i\rceil\}=L_i$. Since $h_i$ is an integer
and $p_i h_i\le c_i g_i-T$, we also have $h_i\le\lfloor (c_i g_i-T)/p_i\rfloor=f_i(g_i).$ 
Every universal chore is allocated, so $\sum_i h_i=H$, proving the
necessary inequality.

Conversely, assign the goods according to $\sigma$. The capacities
$f_i(g_i)$ are nonnegative integers with sum at least $H$. Distribute
the $H$ universal chores to agents with unused capacity, obtaining
integers $0\le h_i\le f_i(g_i)$ with $\sum_i h_i=H$. Assign each
neutral item to an agent in its nonempty set $Z_o$. Each resulting
utility satisfies $c_i g_i-p_i h_i\ge T$, as required.
\end{proof}

\begin{definition}\label{def:mixed-admissible}
A good assignment $\sigma$ is \emph{$T$-admissible} if
$g_i(\sigma)\ge L_i$ for every agent $i$.
\end{definition}

Partition the agents into $N_1=\{i:p_i=1\},$ $N_2^{\mathrm{even}}=\{i:p_i=2,\ c_i\text{ is even}\}$, $N_2^{\mathrm{odd}}=\{i:p_i=2,\ c_i\text{ is odd}\}$. 
Set $\omega_i=2c_i$ for $i\in N_1$ and $\omega_i=c_i$ for $p_i=2$.
Intuitively, $\omega_i$ measures how much assigning one good to agent
$i$ contributes to the doubled chore-absorption capacity, before the
fixed and parity corrections. Define
\begin{align}
&\beta
=2T\cdot|N_1|+T\cdot\bigl(|N_2^{\mathrm{even}}|+|N_2^{\mathrm{odd}}|\bigr)
  +\1[T\text{ is odd}]\cdot|N_2^{\mathrm{even}}|,
  \label{eq:mixed-beta}\\
&\Phi(\sigma)
={\sum}_{o\text{ good}}\omega_{\sigma(o)}
  -{\sum}_{i\in N_2^{\mathrm{odd}}}
      \1[g_i\not\equiv T\pmod2].
  \label{eq:mixed-phi}
\end{align}
Here $\1[\cdot]$ denotes an indicator.

\begin{lemma}[Parity identity]\label{lem:parity}\label{lem:mixed-parity}
    For every good assignment $\sigma$ and every integer $T$, $2\sum_{i=1}^n f_i(g_i)=\Phi(\sigma)-\beta(T).$ 
\end{lemma}
\begin{proof}
For $i\in N_1$, we have $2f_i(g_i)=2c_i g_i-2T$. For $p_i=2$, the
identity $2\lfloor a/2\rfloor=a-\1\{a\text{ is odd}\}$,
valid for every integer $a$, gives
\begin{equation}\label{eq:indicator}
2f_i(g_i)=c_i g_i-T-
\1[c_i g_i-T\text{ is odd}].
\end{equation}
For $i\in N_2^{\mathrm{even}}$, $c_i$ is even, so the indicator equals
$\1[T\text{ is odd}]$. For $i\in N_2^{\mathrm{odd}}$,
$c_i$ is odd, so it equals $\1[g_i\not\equiv T\pmod2]$.
Sum these formulas and use
$\sum_i\omega_i g_i=\sum_{o\text{ good}}\omega_{\sigma(o)}$.
\end{proof}

By Lemmas~\ref{lem:mixed-absorption} and~\ref{lem:mixed-parity}, there
is an allocation of welfare at least $T$ precisely when some
$T$-admissible good assignment satisfies $\Phi(\sigma)\ge=2H+\beta.$
%\begin{equation}\label{eq:threshold}
%\end{equation}

\paragraph{Maximizing $\Phi$ by a weighted perfect matching.}
For every agent $i$, write $U_i=d_i-L_i$. The number
$U_i$ is the maximum number of incident goods that may be assigned to
other agents while agent $i$ still receives at least $L_i$ goods. If
some $U_i<0$, no $T$-admissible assignment exists. Otherwise construct
a graph $\mathcal H$ as follows, with all unspecified edge weights zero.
\begin{enumerate}[leftmargin=*,itemsep=3pt]
    \item \emph{Good vertices.} For every good $o$, add a vertex $x_o$.
    \item \emph{Port vertices.} For every incidence edge $(o,i)$, add a port $q_{i,o}$ and an edge $x_o q_{i,o}$ of weight $\omega_i$.
    \item \emph{Absorber vertices.} For agent $i$, add $U_i$ ordinary absorbers $a_{i,1},\ldots,a_{i,U_i}$ and join every port of $i$ to every ordinary absorber of $i$. If $i\in N_2^{\mathrm{odd}}$ and $L_i\not\equiv T\pmod2$, also add an extra absorber $a_{i,0}$, with no edges to ports. Denote all absorbers of $i$ by $\mathcal A_i$ and make $\mathcal A_i$ a clique.
    \item \emph{Parity vertices.} For each agent $i$, add $z_i$ adjacent to every vertex of $\mathcal A_i$. Give these edges weight $-1$ for $i\in N_2^{\mathrm{odd}}$ and zero otherwise.
    \item \emph{Extra vertices.} If the number of vertices constructed so far is odd, add one more vertex $z_0$. Make $Z=\{z_i:i\in N\}$, together with $z_0$ if present, a clique.
\end{enumerate}
The graph has $O(n(m+1))$ vertices and $O(nm^2+n^2)$ edges, including
the clique on $Z$. It has an even number of vertices. Whether it has a
perfect matching is characterized by Lemma~\ref{lem:mixed-matching}
below.

\begin{example}[The gadget graph $\mathcal H$]\label{ex:gadget}\label{ex:mixed-gadget}
Consider two agents with values in $\{-2,0,3\}$, so that
$p_1=p_2=2$, $c_1=c_2=3$, both agents belong to $N_2^{\mathrm{odd}}$,
and $\omega_1=\omega_2=3$. There are four items:
\begin{center}
\small
\begin{tabular}{lrrrr}
\toprule
 & $o_1$ & $o_2$ & $o_3$ & $h$\\
\midrule
agent 1 & 3 & 3 & 0 & $-2$\\
agent 2 & 3 & $-2$ & 3 & $-2$\\
\bottomrule
\end{tabular}
\end{center}
Thus $o_1,o_2,o_3$ are goods and $h$ is the only universal chore. Agent 1 is adjacent to $o_1,o_2$ and agent 2 to $o_1,o_3$ in $\Gamma$, so $d_1=d_2=2$. The optimum is 3, attained by assigning $\{o_1,o_2,h\}$ to agent 1, of value 4, and $\{o_3\}$ to agent 2, of value 3. For an upper bound, normalize any allocation. With only three goods, one agent receives at most one good and consequently has utility at most 3.

For $T=2$, we have $L_i=\lceil2/3\rceil=1$ and $U_i=1$. Since $L_i\not\equiv T\pmod2$, each agent receives an extra absorber $a_{i,0}$ not adjacent to any port. There are three good vertices, four ports, four absorbers, and the two parity vertices $z_1,z_2$, giving 13 vertices; hence $z_0$ is added. Figure~\ref{fig:mixed-gadget} shows the resulting graph. All good--port edges have weight 3, all edges from $z_i$ to its absorbers have weight $-1$, and all other edges have weight zero.
\end{example}

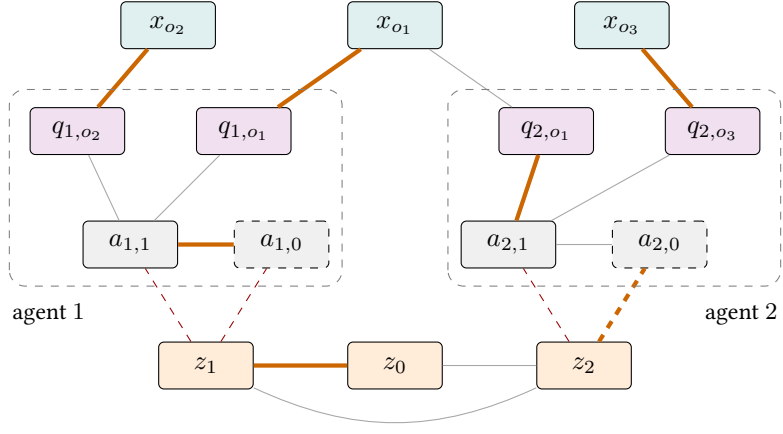
\begin{figure}[t]
\centering
\begin{tikzpicture}[
  every node/.style={draw, rounded corners=2pt, minimum width=1.25cm, minimum height=0.62cm, inner sep=1pt, font=\small},
  good/.style={fill=teal!12},
  port/.style={fill=violet!12},
  absb/.style={fill=gray!12},
  extra/.style={fill=gray!12, dashed},
  par/.style={fill=orange!15},
  e/.style={draw=gray!70, thin},
  m/.style={draw=orange!80!black, line width=1.6pt},
  neg/.style={draw=red!60!black, thin, dashed}
]
\node[good] (o2) at (-3.0, 3.2) {$x_{o_2}$};
\node[good] (o1) at ( 0.0, 3.2) {$x_{o_1}$};
\node[good] (o3) at ( 3.0, 3.2) {$x_{o_3}$};
\node[port] (q12) at (-4.2, 1.8) {$q_{1,o_2}$};
\node[port] (q11) at (-2.0, 1.8) {$q_{1,o_1}$};
\node[port] (q21) at ( 2.0, 1.8) {$q_{2,o_1}$};
\node[port] (q23) at ( 4.2, 1.8) {$q_{2,o_3}$};
\node[absb]  (a11) at (-3.5, 0.3) {$a_{1,1}$};
\node[extra] (a10) at (-1.5, 0.3) {$a_{1,0}$};
\node[absb]  (a21) at ( 1.5, 0.3) {$a_{2,1}$};
\node[extra] (a20) at ( 3.5, 0.3) {$a_{2,0}$};
\node[par] (z1) at (-2.5,-1.3) {$z_1$};
\node[par] (z0) at ( 0.0,-1.3) {$z_0$};
\node[par] (z2) at ( 2.5,-1.3) {$z_2$};
\draw[m] (o2) -- (q12);
\draw[m] (o1) -- (q11);
\draw[e] (o1) -- (q21);
\draw[m] (o3) -- (q23);
\draw[e] (q12) -- (a11);
\draw[e] (q11) -- (a11);
\draw[m] (q21) -- (a21);
\draw[e] (q23) -- (a21);
\draw[m] (a11) -- (a10);
\draw[e] (a21) -- (a20);
\draw[neg] (z1) -- (a11);
\draw[neg] (z1) -- (a10);
\draw[neg] (z2) -- (a21);
\draw[m,dashed] (z2) -- (a20);
\draw[m] (z1) -- (z0);
\draw[e] (z0) -- (z2);
\draw[e] (z1) to[bend right=25] (z2);
\draw[gray, dashed, rounded corners=6pt] (-5.1,-0.25) rectangle (-0.7,2.35);
\draw[gray, dashed, rounded corners=6pt] ( 0.7,-0.25) rectangle ( 5.1,2.35);
\node[draw=none, font=\footnotesize, anchor=north west, minimum width=0pt] at (-5.1,-0.3) {agent $1$};
\node[draw=none, font=\footnotesize, anchor=north east, minimum width=0pt] at ( 5.1,-0.3) {agent $2$};
\end{tikzpicture}
\caption{The graph $\mathcal H$ of Example~\ref{ex:mixed-gadget}. 
 Good--port edges have weight $3$, dashed edges have weight $-1$, and all other edges have weight $0$. Small dashed boxes mark the extra absorbers $a_{1,0},a_{2,0}$. 
Thick edges form the maximum-weight perfect matching of weight $8$.}
\label{fig:mixed-gadget}
\end{figure}

\begin{lemma}[Matching reduction]\label{lem:matching_1}\label{lem:mixed-matching}
The graph $\mathcal H$ has a perfect matching if and only if a $T$-admissible
good assignment exists. Whenever either exists, the maximum weight of
a perfect matching is the maximum of $\Phi(\sigma)$ over all such
assignments.
\end{lemma}
\begin{proof}
First consider the parity mechanism for one agent $i$ in a perfect
matching. Suppose $u_i$ of its ports are matched to absorbers. Ports
reach only $a_{i,1},\ldots,a_{i,U_i}$, so $u_i\le U_i$. The remaining
$\ell_i=|\mathcal A_i|-u_i$ absorbers can be matched only inside the clique
$\mathcal A_i$ or to $z_i$, and $z_i$ covers at most one absorber. Therefore
$z_i$ is matched into $\mathcal A_i$ if and only if $\ell_i$ is odd.

Let $g_i=d_i-u_i$ be the number of ports of agent $i$ matched to goods.
For $i\in N_2^{\mathrm{odd}}$, if $a_{i,0}$ is absent, then
$L_i\equiv T\pmod2$ and
$\ell_i=U_i-u_i=g_i-L_i$. If $a_{i,0}$ is present, then
$L_i\not\equiv T\pmod2$ and
$\ell_i=U_i+1-u_i=g_i-L_i+1$. In both cases,
\begin{equation}\label{eq:odd}
\ell_i\text{ is odd}\quad\Longleftrightarrow\quad
g_i\not\equiv T\pmod2,
\qquad(i\in N_2^{\mathrm{odd}}).
\end{equation}

\emph{From matchings to assignments.}
Let $\mathcal M$ be a perfect matching. Every $x_o$ is adjacent only
to ports $q_{i,o}$ with $i\in G_o$. Set $\sigma(o)=i$ when $x_o$ is
matched to $q_{i,o}$. This defines a good assignment. Every port is
adjacent only to its good and its agent's ordinary absorbers, so the
ports of agent $i$ matched to goods correspond exactly to
$\sigma^{-1}(i)$. Consequently, $g_i(\sigma)=d_i-u_i\ge d_i-U_i=L_i$, 
and $\sigma$ is $T$-admissible. The weight of $\mathcal M$ is
$\sum_o\omega_{\sigma(o)}$ minus the number of
$i\in N_2^{\mathrm{odd}}$ whose vertex $z_i$ is matched into $\mathcal A_i$.
By~\eqref{eq:odd}, this weight is $\Phi(\sigma)$.

\emph{From assignments to matchings.}
Let $\sigma$ be $T$-admissible. Match each $x_o$ to $q_{\sigma(o),o}$.
For each agent $i$, match its $u_i=d_i-g_i\le U_i$ remaining ports to
distinct ordinary absorbers. If the number $\ell_i$ of remaining
absorbers is odd, match one of them to $z_i$; pair the rest inside
$\mathcal A_i$. Every vertex outside $Z$ is now covered. Since every matching
edge covers two vertices and the number of vertices in $\mathcal H$ is even, the number of
uncovered vertices in $Z$ is even. They can be paired because $Z$ is
a clique. The resulting perfect matching has weight $\Phi(\sigma)$
by~\eqref{eq:odd}. The two constructions prove both
claims.
\end{proof}

By Lemma~\ref{lem:mixed-absorption} and \ref{lem:mixed-matching}, the proof of Theorem~\ref{thm:mixed-exact} is followed. 

\paragraph{Continuation of Example~\ref{ex:mixed-gadget}.} The highlighted matching in Figure~\ref{fig:mixed-gadget} corresponds to $\sigma(o_1)=1$, $\sigma(o_2)=1$, and $\sigma(o_3)=2$. Agent 1 uses both ports, so its two absorbers are matched to each other. Since $g_1=2\equiv T\pmod2$, it incurs no penalty. Agent 2 does not use $q_{2,o_1}$, which is matched to $a_{2,1}$; the leftover absorber $a_{2,0}$ is matched to $z_2$ at cost $-1$, reflecting $g_2=1\not\equiv T\pmod2$. Finally, $z_0$ is matched to $z_1$. The weight is $3+3+3-1=8=\Phi_2(\sigma)$. Any admissible assignment has goods counts $(2,1)$ or $(1,2)$, so this is the maximum. As $\beta(2)=2|N_2^{\mathrm{odd}}|=4$ and $\alpha(2)=2H+\beta(2)=6\le8$, the threshold $T=2$ is feasible.

\section[Hardness for Ternary Submodular Valuations]{Hardness for $\{-2,0,c\}$-Submodular Valuations}
\label{sec:submod}
%==================================================================================

In this section, we consider $\{-2,0,c\}$-submodular valuations. In contrast to the additive case, we show that $\MEW$ is strongly $\np$-hard for this class of valuations. %Moreover, it is impossible to approximate $\MEW$ within an additive error strictly smaller than $1$, unless $\p=\np$.

\begin{theorem}
\label{thm:submod}\label{thm:submod-odd}
Let $c\geq1$ be a fixed odd integer. For polynomial-time computable
$\{-2,0,c\}$-submodular valuations, it is $\np$-hard to distinguish
instances with $\OPT=T$ from instances with $\OPT=T-1$, where
$T>1$ is a polynomially bounded integer. This holds even when all
agents but one are additive. Consequently, unless $\p=\np$, this
class admits neither an FPTAS nor a polynomial-time algorithm returning
an allocation with additive error strictly smaller than one.

%In the value-oracle model, exact optimization and finding an allocation with additive error strictly smaller than one require $2^{\Omega(n)}=2^{\Omega(\sqrt m)}$ queries on a family with $n$ agents and $m=\Theta(n^2)$ items, again with only one nonadditive agent. This lower bound also holds in expected query complexity for randomized algorithms that succeed with probability at least $2/3$ on every instance.
\end{theorem}

\begin{lemma}\label{lem:matroid-val}
If $\rho$ is a matroid rank function, then
$v(S)=(c+2)\rho(S)-2|S|$ is $\{-2,c\}$-submodular, and $ v(S)=c\rho(S)-2\bigl(|S|-\rho(S)\bigr).$
\end{lemma}

Thus an agent with such a valuation perceives the items that raise
its rank as goods of value $c$ and all other items as chores of
value $-2$. Such valuations are {\em order-neutral}: every ordering of a
bundle $S$ has exactly $\rho(S)$ marginals of value $c$ and
$|S|-\rho(S)$ marginals of value $-2$. For odd positive $c$, every
$\{-2,0,c\}$-submodular valuation is order-neutral as well. Indeed,
the six unordered pair sums
$-4,-2,0,c-2,c,2c$ are distinct, so swapping adjacent items in an
ordering preserves the unordered pair of their marginals, and
therefore the entire marginal multiset. This is the argument of
\citet[Proposition~4.8]{FVZ25}. Hence order-neutrality alone does not
yield tractability when chores have magnitude two.

\subsection{A sparse paving encoding of CLIQUE}
\label{subsec:submod-encoding}

Let $k\geq 2$ and let
$E=\{e_1,e'_1,\ldots,e_k,e'_k\}$ be a set of $2k$ elements,
partitioned into the pairs $P_u=\{e_u,e'_u\}$ for $u\in[k]$.
For $J\subseteq[k]$, write $P_J=\bigcup_{u\in J}P_u$.
Fix an integer $s$ with $1\leq s<k$, and put $r=2s$.
In the reduction, the elements of $E$ are the pair items and $s$ is
the target clique size. Restricting CLIQUE to $1\leq s<k$ preserves
NP-hardness, since its other endpoints are decidable directly.

\begin{lemma}
\label{lem:sparse-paving}\label{lem:submod-sparse-paving}
Let $\mathcal F\subseteq\{P_J:|J|=s\}$. Then
\[
 \rho_{\mathcal F}(S)=
 \begin{cases}
 |S|, & |S|<r,\\
 r-1, & S\in\mathcal F,\\
 r,   & \text{otherwise}
 \end{cases}
\]
is the rank function of a matroid of rank $r$ on $E$. A set $P_J$
with $|J|=s$ is independent if and only if $P_J\notin\mathcal F$.
\end{lemma}

Given a graph $\mathcal G$ on vertex set $[k]$, let
\[
 \mathcal F_{\mathcal G}=\{P_J:|J|=s,\ J\text{ is not a clique of }\mathcal G\}.
\]
Membership in $\mathcal F_{\mathcal G}$ is decidable in polynomial time: check
that the set is a union of $s$ whole pairs and then check the edges
among their indices. Thus $\rho_{\mathcal F_{\mathcal G}}$ is polynomial-time
computable, and a union of $s$ pairs is independent exactly when
its indices form a clique.

\subsection{The reduction}
\label{subsec:submod-pair-reduction}

Let $c$ be a fixed odd positive integer and let $(\mathcal G,s)$ be a CLIQUE
instance on $k$ vertices with $1\leq s<k$. Write $M_{\mathcal G}$ for the
matroid in Lemma~\ref{lem:sparse-paving} with
$\mathcal F=\mathcal F_{\mathcal G}$, and set: $d=2k$, $d_0=2k-2s$, $H=c(k-s)$, $T=cd.$

\paragraph{Items.} The $2k$ pair items $E=\{e_1,e'_1,\ldots,e_k,e'_k\}$,
where $P_u=\{e_u,e'_u\}$ is the pair of vertex $u$; for each
$u\in[k]$ a set $D_u$ of $d$ private dummy items; a set $D_0$ of
$d_0$ private dummy items; and a set $C$ of $H$ chores. These sets
are pairwise disjoint.

\paragraph{Agents.} There are $k+1$ agents, with rank functions
as follows:
\begin{itemize}%[nosep]
\item the \emph{collector} $0$, with
$\rho_0(S)=\rho_{\mathcal F_{\mathcal G}}(S\cap E)+|S\cap D_0|$;
\item for each $u\in[k]$ a \emph{pair agent} $u$, with
$\rho_u(S)=|S\cap(P_u\cup D_u)|$.
\end{itemize}
These are matroid rank functions: the collector takes a direct sum
with the free matroid on $D_0$, and each pair agent uses a free
matroid on $P_u\cup D_u$; all unspecified items are loops.

\paragraph{Valuations.} For every agent $i$, define $ v_i(S)=(c+2)\rho_i(S)-2|S|.$ 
Each pair agent is additive: it values every item of $P_u\cup D_u$
at $c$ and every other item at $-2$. Only the collector can be
non-additive. Every chore $h\in C$ is a loop for every agent,
that is, $\rho_i(S\cup\{h\})=\rho_i(S)$ for every $S$; hence
its marginal value is $-2$ regardless of the bundle.

\begin{claim}
\label{clm:submod-size}
For fixed $c$, the construction has polynomial size and all its
valuations are computable in polynomial time.
\end{claim}

\begin{lemma}[Completeness]
\label{lem:submod-completeness}
If $\mathcal G$ has a clique of size $s$, some allocation has welfare at
least $T$.
\end{lemma}

\begin{proof}
Let $J$ be such a clique. Give the collector $D_0\cup P_J$.
This set is independent for $\rho_0$ and has size
$d_0+2s=d$, so its value is $cd=T$. Give each $u\in J$ the set
$D_u$, of value $cd$. Give each $u\notin J$ the set
$D_u\cup P_u$ and $c$ chores. Its value is
$c(d+2)-2c=cd$. This uses exactly $c(k-s)=H$ chores and allocates
every item.
\end{proof}

\begin{lemma}[Soundness]
\label{lem:submod-soundness}
If some allocation has welfare at least $T$, then $\mathcal G$ has a clique
of size $s$.
\end{lemma}

\begin{proof}
Let $A$ be such an allocation. For agent $i$, let $N_i$ be the set
of nonloops of $\rho_i$. Since $r\geq 2$, every pair item is a
nonloop of the collector. Hence $R=\left|\bigcup_iN_i\right|=2k+kd+d_0,$ and the total number of items is $R+H$.

\emph{Tightness.}
For every $i$, $\rho_i(A_i)\leq|A_i\cap N_i|$. Since bundles are
disjoint, $\sum_i\rho_i(A_i)\leq R$. Because every item is allocated,
\begin{align*}
 \sum_i v_i(A_i)
 \leq(c+2)R-2(R+H)
 =cR-2H=c(kd+d)=(k+1)T.
\end{align*}
Every agent has value at least $t$, so every value equals $t$ and
$\sum_i\rho_i(A_i)=R$. Equality forces
$\rho_i(A_i)=|A_i\cap N_i|$ for every $i$, and every item of
$\bigcup_iN_i$ must be assigned to an agent for whom it is a
nonloop. Thus each $A_i\cap N_i$ is independent; each $D_u$ goes to
$u$; and $D_0$ goes to the collector. The only items that can be
loops of their holders are the universal chores in $C$.

\emph{Pair agents keep whole pairs or nothing.}
For $u\in[k]$, put $x_u=|A_u\cap P_u|$. Then
$\rho_u(A_u)=d+x_u$, and $v_u(A_u)=cd$ gives $c\bigl(\rho_u(A_u)-d\bigr)
   =2\bigl(|A_u|-\rho_u(A_u)\bigr).$ 
The right-hand side is even. Since $c$ is odd, $x_u$ is even,
and therefore $x_u\in\{0,2\}$.

\emph{The collector obtains a clique.}
Each pair item is a nonloop only for its pair agent and the
collector. Consequently every pair not kept by its pair agent goes
entirely to the collector, and $A_0\cap E=P_J$ for some
$J\subseteq[k]$. This set is independent in $M_{\mathcal G}$.
Since $v_0(A_0)=cd$ and $|A_0|\geq\rho_0(A_0)$, $cd=c\rho_0(A_0)-2\bigl(|A_0|-\rho_0(A_0)\bigr)
       \leq c\rho_0(A_0).$ 
Thus $\rho_0(A_0)\geq d$. But
$\rho_0(A_0)\leq d_0+2s=d$, so equality holds. Independence
gives $\rho_0(A_0)=d_0+2|J|$, and hence $|J|=s$.
By Lemma~\ref{lem:sparse-paving}, $J$ is a clique. These
equalities also show that the collector receives no chore.
\end{proof}

\begin{proof}[Proof of Theorem~\ref{thm:submod}]
Claim~\ref{clm:submod-size} and
Lemmas~\ref{lem:submod-completeness}--\ref{lem:submod-soundness}
give a polynomial reduction from CLIQUE for every fixed odd $c$.
All but one agent are additive, and Lemma~\ref{lem:matroid-val}
verifies the required marginals and submodularity. In fact the
construction has the exact gap: 
\[
 \OPT=
 \begin{cases}
 T, & \mathcal G\text{ has a clique of size }s,\\
 T-1, & \text{otherwise}.
 \end{cases}
\]
The total-welfare bound in the soundness proof gives $\OPT\leq T$
for every instance. Completeness gives equality on YES-instances;
soundness and integrality give $\OPT\leq T-1$ on NO-instances.
It remains to construct an allocation of welfare at least $T-1$
regardless of $\mathcal G$.

Choose $s-1$ whole pairs and two further distinct pairs; this is
possible because $s<k$. Give the collector $D_0$, those $s-1$
whole pairs, and one item from each of the two further pairs. Its
$2s$ pair items are not a union of whole pairs and therefore do
not belong to $\mathcal F_{\mathcal G}$. They are independent, and the
collector's value is $T$. Give every pair agent its private dummies.
Each agent retaining its whole pair receives $c$ chores. The two
agents retaining exactly one pair item receive respectively
$(c-1)/2$ and $(c+1)/2$ chores; their values are $T+1$ and $T-1$.
The agents retaining no pair item receive no chore and have value
$T$. The number of chores used is
\[
 c(k-s-1)+\frac{c-1}{2}+\frac{c+1}{2}=c(k-s)=H.
\]
This allocates every item and proves the lower bound. Since
$T=2ck\geq4$, both possible optima are strictly positive.
For the explicitly represented family, an allocation is a
polynomial-size certificate whose welfare is polynomial-time
computable, so threshold decision is $\np$-complete as well.

An algorithm with additive error strictly smaller than one would
return integer welfare greater than $T-1$ on every YES-instance,
and hence welfare $T$. No NO-instance admits such an allocation.
For an assumed FPTAS, choose $\varepsilon=1/(2t)$. Its guarantee
of $\OPT/(1+\varepsilon)\geq(1-\varepsilon)\OPT$ gives welfare
at least $T-1/2$ on a YES-instance, again forcing welfare $T$.
Since $1/\varepsilon=4ck$ is polynomially bounded for fixed $c$,
either algorithm would decide CLIQUE in polynomial time.
\end{proof}

\begin{example}[The reduction from CLIQUE]
\label{ex:clique-reduction}\label{ex:submod-clique}
Let $\mathcal G$ be the path on vertices $\{1,2,3\}$ with edges
$\{1,2\}$ and $\{2,3\}$, let $s=2$, and let $c=3$. Then $r=4$, $d=6$, $d_0=2$, $H=3$, $t=18.$ 
There are six pair items, eighteen private dummy items for the
pair agents, two dummy items for the collector, and three chores.
The only nonclique pair of vertices is $\{1,3\}$, so
$\mathcal F_{\mathcal G}=\{P_1\cup P_3\}$. Every set of at most three pair
items is independent, $P_1\cup P_3$ has rank three, and every other
set of at least four pair items has rank four. In particular,
$P_1\cup P_2$ and $P_2\cup P_3$ are independent.

For the clique $J=\{1,2\}$, completeness gives the following
allocation, shown also in Figure~\ref{fig:clique-reduction}. Here
$v_i(A_i)=5\rho_i(A_i)-2|A_i|$.
\begin{center}
\begin{tabular}{lllll}
\toprule
Agent & Bundle $A_i$ & $\rho_i(A_i)$ & $|A_i|$ & $v_i(A_i)$\\
\midrule
Collector & $D_0\cup P_1\cup P_2$ & $2+4=6$ & $6$ & $18$\\
Pair agent $1$ & $D_1$ & $6$ & $6$ & $18$\\
Pair agent $2$ & $D_2$ & $6$ & $6$ & $18$\\
Pair agent $3$ & $D_3\cup P_3\cup C$ & $8$ & $11$ & $18$\\
\bottomrule
\end{tabular}
\end{center}
Using $P_1\cup P_3$ for the collector instead would reduce its
value to $5\cdot5-2\cdot6=13$.

For the soundness argument, replace $\mathcal G$ by the empty graph on the
same vertices. There are twenty-six nonchore items, so total welfare
is at most $3\cdot26-2\cdot3=72=4t$. If every agent had value at
least eighteen, all values would equal eighteen and every nonchore
item would contribute to the rank of its holder. A pair agent
keeping exactly one pair item would have rank seven and odd value
$35-2|A_u|$, so it could not attain eighteen. Thus the collector
would receive two whole pairs forming an independent set, which is
impossible in this instance. In fact the optimum is seventeen:
give the collector $D_0\cup P_1\cup\{e_2,e_3\}$, give pair agent
$1$ only $D_1$, and give pair agents $2,3$ their dummies and their
remaining pair item, together with one and two chores respectively.
Their values are eighteen, eighteen, nineteen, and seventeen.
\end{example}

\begin{figure}[!ht]
\centering
\begin{tikzpicture}[
  every node/.style={font=\small},
  vtx/.style={draw, circle, fill=gray!15, minimum size=0.55cm, inner sep=0pt},
  item/.style={draw, rounded corners=2pt, fill=teal!12, minimum width=0.8cm, minimum height=0.55cm, inner sep=1pt},
  agent/.style={draw, rounded corners=2pt, fill=violet!12, minimum width=1.6cm, minimum height=0.9cm, align=center, inner sep=2pt},
  coll/.style={draw, rounded corners=2pt, fill=orange!15, minimum width=2.4cm, minimum height=0.9cm, align=center, inner sep=2pt},
  e/.style={draw=gray!70, thin},
  m/.style={draw=orange!80!black, line width=1.6pt}
]
\begin{scope}[shift={(-7.2,0.9)}]
\node[vtx] (v1) at (0,1.4) {$1$};
\node[vtx] (v2) at (0,0) {$2$};
\node[vtx] (v3) at (0,-1.4) {$3$};
\draw[thick] (v1) -- (v2) -- (v3);
\node at (0,-2.2) {$\mathcal G$, $s=2$};
\end{scope}
\node[agent] (p1) at (-3.6,2.4) {pair agent $1$\\[-1pt]{\scriptsize $+D_1$}};
\node[agent] (p2) at ( 0.0,2.4) {pair agent $2$\\[-1pt]{\scriptsize $+D_2$}};
\node[agent] (p3) at ( 3.6,2.4) {pair agent $3$\\[-1pt]{\scriptsize $+D_3$, $+3$ chores}};
\node[item] (e1)  at (-4.2,0.8) {$e_1$};
\node[item] (e1p) at (-3.0,0.8) {$e_1'$};
\node[item] (e2)  at (-0.6,0.8) {$e_2$};
\node[item] (e2p) at ( 0.6,0.8) {$e_2'$};
\node[item] (e3)  at ( 3.0,0.8) {$e_3$};
\node[item] (e3p) at ( 4.2,0.8) {$e_3'$};
\node[coll] (c) at (0,-1.2) {collector\\[-1pt]{\scriptsize $+D_0$, hidden matroid}};
\draw[e] (p1) -- (e1);  \draw[e] (p1) -- (e1p);
\draw[e] (p2) -- (e2);  \draw[e] (p2) -- (e2p);
\draw[m] (p3) -- (e3);  \draw[m] (p3) -- (e3p);
\draw[m] (c) -- (e1);  \draw[m] (c) -- (e1p);
\draw[m] (c) -- (e2);  \draw[m] (c) -- (e2p);
\draw[e] (c) -- (e3);  \draw[e] (c) -- (e3p);
\end{tikzpicture}
\caption{The instance of Example~\ref{ex:clique-reduction} for the path $\mathcal G$ with edges $\{1,2\}$ and $\{2,3\}$ and $s=2$. Each pair item $e_u,e_u'$ can be counted either by pair agent $u$ (value $c$) or by the collector (through the hidden sparse paving matroid, in which $P_1\cup P_3$ is dependent). Thick edges show the allocation for the clique $\{1,2\}$: the collector receives $P_1\cup P_2$ and its dummies $D_0$, pair agents $1$ and $2$ keep only their dummies, and pair agent $3$ keeps $P_3$, its dummies, and all $H=3$ chores. Every agent has value $t=18$.
}
\label{fig:clique-reduction}\label{fig:submod-clique}
\end{figure}
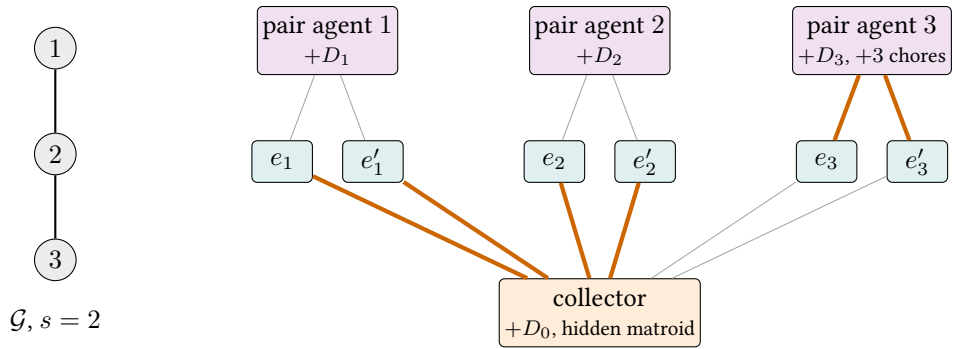

%==================================================================================

%==================================================================================
\section[Inapproximability of Conflict-Pair Valuations]{Approximability of $\{-1,0,1\}$-Submodularity}\label{sec:m101}
%==================================================================================

%We investigate the case with $\{-1,0,1\}$-submodular valuations and prove that there is no approximation algorithm with finite approximation ratio for $\MEW$, unless $\p=\np$. This holds even for conflict-pair valuations, a sub-class of submodular valuations.

%\subsection{Conflict-pair valuations}\label{sec:conflict}

For $\{-1,0,c\}$-submodular valuations with an integer $c\geq 2$, every valuation is order-neutral, and maximum egalitarian welfare is polynomial-time computable \citep{FVZ25,CVZ23b}. For $c=1$, the pairwise sums of marginals collide, since $-1+1=0+0$, and non-order-neutral valuations exist. The basic example consists of two items $x,y$ with $v(\varnothing)=0$, $v(\{x\})=1$, $v(\{y\})=0$, $v(\{x,y\})=0.$ 
Adding $x$ and then $y$ yields marginals $(1,-1)$, whereas adding $y$ and then $x$ yields $(0,0)$. This is a conflict-pair valuation as in Definition~\ref{def:conflict-pair}. By Lemma~\ref{lem:conflict-pair-submod}, all conflict-pair valuations are $\{-1,0,1\}$-submodular.

\begin{theorem}\label{thm:hard_101}\label{thm:conflict-hardness}
Deciding whether $\OPT\geq 1$ is NP-complete for explicitly represented conflict-pair valuations, even when each agent has at most one conflict pair and values at most two items at $1$. NP-hardness holds under the additional promise that $\OPT\in\{0,1\}$. 
Unless $\p=\np$, there is no polynomial-time multiplicative approximation algorithm with a finite approximation ratio, even under the promise $\OPT=1$, when an allocation must be returned. There is also no polynomial-time algorithm that always returns an allocation with welfare at least $\OPT-\varepsilon$, for any fixed $0\leq\varepsilon<1$.
\end{theorem}

We reduce from \textsc{3-SAT} \citep{GareyJohnson1979}. Let $\varphi$ be a formula with variables $x_1,\ldots,x_s$ and nonempty clauses $C_1,\ldots,C_t$, each containing at most three literal occurrences. In this section, the capital letters $s$ and $t$ denote the numbers of variables and clauses of this source formula; they do not denote the agent and item sets of the allocation instance. We write $n$ and $m$ for the numbers of agents and items in the constructed instance. Let $o_v$ be the number of occurrences of $x_v$, and put $L_v=\max\{2,o_v\}.$ 

\paragraph{Variable cycles.}
For every variable $x_v$, create agents $a_{v,1},\ldots,a_{v,L_v}$ and goods $g_{v,1},\ldots,g_{v,L_v}$. Agent $a_{v,k}$ has singleton value $1$ for $g_{v,k}$ and $g_{v,k+1}$, with indices interpreted modulo $L_v$. Call these its left and right goods, respectively. Since $L_v\geq 2$, the two goods are distinct. The bipartite graph of these agents and their positive goods is a single even cycle, with exactly two perfect matchings: every agent takes its left good, read as $x_v=\mathrm{true}$, or every agent takes its right good, read as $x_v=\mathrm{false}$. Indeed, the choice of either edge at one agent forces the choices around the entire cycle. This also holds for $L_v=2$, when the graph is a four-cycle.

\paragraph{Clause items.}
For each clause $C$, create an item $y_C$. Assign the occurrences of $x_v$ injectively to agents of the $v$-cycle. This is possible because $L_v\geq o_v$, and each agent is assigned at most one occurrence in the whole formula. If an occurrence of $x_v$ in clause $C$ is assigned to $a_{v,k}$, give this agent singleton value $0$ for $y_C$. Its conflict pair is
\[
\begin{cases}
(g_{v,k+1},y_C),&\text{if the occurrence is positive},\\
(g_{v,k},y_C),&\text{if the occurrence is negative}.
\end{cases}
\]
Thus the conflicting good is precisely the one chosen when the literal is false. An agent without an assigned occurrence has no conflict pair. All unspecified singleton values are $-1$. The valuation of each agent is the sum of its singleton values, minus one if its bundle contains both items of its conflict pair.

\begin{lemma}\label{lem:conflict-construction}
The construction has the following properties.
\begin{enumerate}[label=(\alph*)]
\item Every valuation is a conflict-pair valuation, with exactly two singleton values equal to $1$, at most one singleton value equal to $0$, and at most one conflict pair.
\item The instance has $n=\sum_{v=1}^{s}L_v\leq 3t+2s$ agents, and $m=n+t\leq 4t+2s$ items. In particular, it has exactly $n$ cycle goods. It can be constructed in time polynomial in $s+t$.
\end{enumerate}
\end{lemma}

\begin{proof}[Proof of Theorem~\ref{thm:conflict-hardness}]
Membership in $\np$ follows by evaluating the bundles of a proposed complete allocation using the explicit singleton values and conflict pairs. We prove correctness of the reduction and then establish the promise and approximation claims.

\paragraph{Completeness.} Let $\sigma$ satisfy $\varphi$. Match each variable cycle according to $\sigma$, and assign each clause item $y_C$ to an agent representing a literal of $C$ that is true under $\sigma$. Every such agent receives its chosen good and at most one clause item, because it represents at most one occurrence. Its chosen good is not paired with that clause item, so its value is $1$. Every other agent receives exactly one of its positive goods and also has value $1$. All items are assigned, and hence $\OPT\geq1$.

\paragraph{ Soundness.} Suppose a complete allocation $A$ gives every agent value at least $1$. For each agent $i$, let $B_i$ be its set of two positive goods. By the definition of the valuation, $v_i(A_i)\leq |A_i\cap B_i|,$ 
with equality if and only if $A_i$ contains neither an item with singleton value $-1$ for $i$ nor both items of a conflict pair of $i$. The sets $A_i\cap B_i$ are disjoint subsets of the $n$ cycle goods. Therefore $n\leq{\sum}_i v_i(A_i) \leq{\sum}_i|A_i\cap B_i| \leq n.$
Equality holds throughout. Every agent has value exactly $1$, receives exactly one of its positive goods, receives no item that it values at $-1$, and holds no complete conflict pair. Every cycle good is assigned to an agent that values it positively. Thus each variable cycle is assigned according to one of its two perfect matchings, defining a Boolean assignment $\sigma$. Since the allocation is complete, each clause item $y_C$ has an owner. It is valued at $-1$ by everyone except agents representing occurrences in $C$, so its owner is one of those occurrence agents. That agent does not hold the good paired with $y_C$. By construction, its corresponding literal is true under $\sigma$. Every clause is therefore satisfied.~

\paragraph{The promise $\OPT\in\{0,1\}$.}
For every allocation, the last two inequalities in \eqref{eq:conflict-counting} hold without any assumption on its welfare. Consequently, the sum of utilities is at most $n$, and $\OPT\leq1$.

For the reverse bound, choose any Boolean assignment, without requiring it to satisfy the formula, and match each cycle accordingly. For each nonempty clause $C$, assign $y_C$ to an agent representing any one of its occurrences. The selected recipient agents are distinct: an agent represents at most one occurrence, belonging to just one clause. Each agent therefore receives one positive good and at most one clause item of singleton value zero. Its utility is $1$ if no conflict is activated and $0$ otherwise. This is a complete allocation of nonnegative welfare, proving $\OPT\geq0$. Since utilities are integral, every constructed instance satisfies $\OPT\in\{0,1\}$. The completeness and soundness arguments show that its optimum is $1$ exactly when $\varphi$ is satisfiable, and otherwise is $0$.

\paragraph{Inapproximability.}
Suppose a polynomial-time algorithm returns an allocation of welfare at least $\OPT/\rho$ whenever $\OPT=1$, where $1\leq\rho<\infty$; the ratio may even depend on the input size. Run it on an instance of the reduction and accept precisely when its output is a valid complete allocation of strictly positive welfare. On a satisfiable instance, the promised guarantee is strictly positive, so integrality forces welfare at least $1$. On an unsatisfiable instance, no valid allocation has positive welfare, regardless of the algorithm's behavior outside its promise. This would decide \textsc{3-SAT} in polynomial time. Here polynomial time has its usual total running-time meaning; alternatively, stop the algorithm at its polynomial bound for promised inputs and treat a timeout as rejection. The result concerns finding an allocation: under the promise $\OPT=1$, merely reporting the optimum value is trivial. Similarly, an additive guarantee $\EW(A)\geq\OPT-\varepsilon$ with $\varepsilon<1$ forces positive welfare whenever $\OPT=1$, and gives the same contradiction.
\end{proof}

\begin{example}[The reduction from \textsc{3-SAT}]\label{ex:sat-reduction}
We illustrate Theorem~\ref{thm:conflict-hardness} with
\[
\varphi=(x_1\lor\neg x_2)\land(\neg x_1\lor x_2).
\]
Write $C_1,C_2$ for the clauses and $y_1,y_2$ for their items. Both variables occur twice, so $L_1=L_2=2$, and the instance has four agents and six items. Assign the occurrences as follows:
\begin{center}
\begin{tabular}{@{}lllll@{}}
\toprule
Occurrence & Agent & Left good & Right good & Conflict pair\\
\midrule
$x_1$ in $C_1$ & $a_{1,1}$ & $g_{1,1}$ & $g_{1,2}$ & $(g_{1,2},y_1)$\\
$\neg x_1$ in $C_2$ & $a_{1,2}$ & $g_{1,2}$ & $g_{1,1}$ & $(g_{1,2},y_2)$\\
$\neg x_2$ in $C_1$ & $a_{2,1}$ & $g_{2,1}$ & $g_{2,2}$ & $(g_{2,1},y_1)$\\
$x_2$ in $C_2$ & $a_{2,2}$ & $g_{2,2}$ & $g_{2,1}$ & $(g_{2,1},y_2)$\\
\bottomrule
\end{tabular}
\end{center}
The resulting singleton values are
\begin{center}
\begin{tabular}{@{}lrrrrrr@{}}
\toprule
 & $g_{1,1}$ & $g_{1,2}$ & $g_{2,1}$ & $g_{2,2}$ & $y_1$ & $y_2$\\
\midrule
$a_{1,1}$ & $1$ & $1$ & $-1$ & $-1$ & $0$ & $-1$\\
$a_{1,2}$ & $1$ & $1$ & $-1$ & $-1$ & $-1$ & $0$\\
$a_{2,1}$ & $-1$ & $-1$ & $1$ & $1$ & $0$ & $-1$\\
$a_{2,2}$ & $-1$ & $-1$ & $1$ & $1$ & $-1$ & $0$\\
\bottomrule
\end{tabular}
\end{center}
The assignment $x_1=x_2=\mathrm{true}$ satisfies $C_1$ through $x_1$ and $C_2$ through $x_2$. Let every agent take its left good, assign $y_1$ to $a_{1,1}$, and assign $y_2$ to $a_{2,2}$:
\[
\begin{aligned}
A_{a_{1,1}}&=\{g_{1,1},y_1\},&
A_{a_{1,2}}&=\{g_{1,2}\},\\
A_{a_{2,1}}&=\{g_{2,1}\},&
A_{a_{2,2}}&=\{g_{2,2},y_2\}.
\end{aligned}
\]
Agent $a_{1,1}$ does not hold $g_{1,2}$, the partner of $y_1$, and agent $a_{2,2}$ does not hold $g_{2,1}$, the partner of $y_2$. Every agent thus has value exactly $1$. By contrast, moving $y_1$ to $a_{2,1}$ would activate the conflict pair $(g_{2,1},y_1)$ and reduce this agent's value to $0$, reflecting that the literal $\neg x_2$ is false.
\end{example}

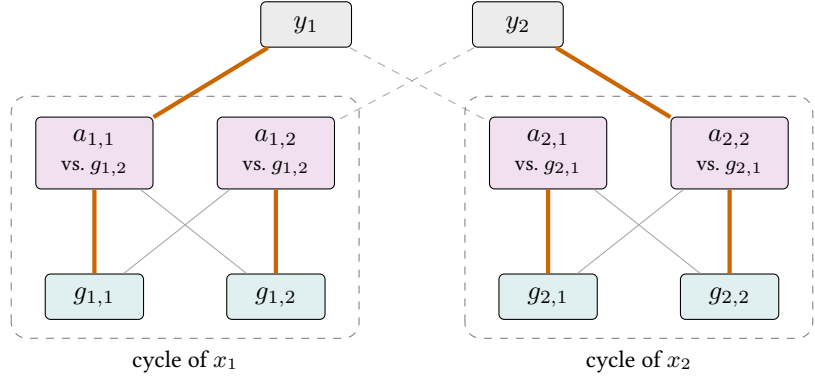
\begin{figure}[!ht]
\centering
\begin{tikzpicture}[
  every node/.style={draw, rounded corners=2pt, inner sep=2pt, font=\small},
  agent/.style={fill=violet!12, minimum width=1.55cm, minimum height=0.95cm, align=center},
  good/.style={fill=teal!12, minimum width=1.3cm, minimum height=0.6cm},
  citem/.style={fill=gray!15, minimum width=1.2cm, minimum height=0.6cm},
  e/.style={draw=gray!70, thin},
  zeroedge/.style={draw=gray!70, thin, dashed},
  m/.style={draw=orange!80!black, line width=1.6pt}
]
\node[citem] (y1) at (-1.4, 3.6) {$y_1$};
\node[citem] (y2) at ( 1.4, 3.6) {$y_2$};
\node[agent] (a11) at (-4.2, 1.9) {$a_{1,1}$\\[-1pt]{\scriptsize vs.\ $g_{1,2}$}};
\node[agent] (a12) at (-1.8, 1.9) {$a_{1,2}$\\[-1pt]{\scriptsize vs.\ $g_{1,2}$}};
\node[agent] (a21) at ( 1.8, 1.9) {$a_{2,1}$\\[-1pt]{\scriptsize vs.\ $g_{2,1}$}};
\node[agent] (a22) at ( 4.2, 1.9) {$a_{2,2}$\\[-1pt]{\scriptsize vs.\ $g_{2,1}$}};
\node[good] (g11) at (-4.2, 0) {$g_{1,1}$};
\node[good] (g12) at (-1.8, 0) {$g_{1,2}$};
\node[good] (g21) at ( 1.8, 0) {$g_{2,1}$};
\node[good] (g22) at ( 4.2, 0) {$g_{2,2}$};
\draw[m] (a11) -- (g11);
\draw[e] (a11) -- (g12);
\draw[m] (a12) -- (g12);
\draw[e] (a12) -- (g11);
\draw[m] (a21) -- (g21);
\draw[e] (a21) -- (g22);
\draw[m] (a22) -- (g22);
\draw[e] (a22) -- (g21);
\draw[m] (y1) -- (a11);
\draw[zeroedge] (y1) -- (a21);
\draw[zeroedge] (y2) -- (a12);
\draw[m] (y2) -- (a22);
\draw[gray, dashed, rounded corners=6pt] (-5.3,-0.55) rectangle (-0.7,2.65);
\draw[gray, dashed, rounded corners=6pt] ( 0.7,-0.55) rectangle ( 5.3,2.65);
\node[draw=none, font=\footnotesize] at (-3.0,-0.85) {cycle of $x_1$};
\node[draw=none, font=\footnotesize] at ( 3.0,-0.85) {cycle of $x_2$};
\end{tikzpicture}
\caption{The instance of Example~\ref{ex:sat-reduction} for $\varphi=(x_1\vee\neg x_2)\wedge(\neg x_1\vee x_2)$. Each dashed box is the cycle of a variable: vertical edges (every agent takes its left good) encode \emph{true}, diagonal edges encode \emph{false}. Solid thin edges join agents to goods they value at $1$, dashed edges join agents to clause items they value at $0$, and ``vs.\ $g$'' names the good forming a conflict pair with the agent's clause item. Thick edges show the allocation for $x_1=x_2=\mathrm{true}$, in which every agent has value $1$.}
\label{fig:sat-reduction}\label{fig:conflict-cycles}
\end{figure}
%==================================================================================

%==================================================================================
\section{Discussion and Open Problems}\label{sec:discussion}
%==================================================================================

We have located the boundary between tractable and intractable max-min allocation for valuations whose item values or marginals come from a small set. For positive bi-valued additive instances with values in $\{p,q\}$, complexity depends only on the reduced denominator $p$: the problem is polynomial-time solvable for $p\in\{1,2\}$, and NP-hard to approximate within a constant factor for $p\ge3$. %This mirrors the dichotomy known for Nash social welfare, but our algorithm for $p=2$ is considerably simpler, resting on one structural fact. For $p=2$, each agent's marginal savings alternate between two consecutive values and decompose into pair slots that a single maximum-weight matching captures exactly. 
For $p\ge3$, longer flat runs in the savings encode exact cover. The same lens resolves the conjectured tractability of $\{-2,0,c\}$-valuations for mixed manna, while our hardness and query lower bounds for $\{-2,0,c\}$- and $\{-1,0,1\}$-submodular valuations show that this tractability relies on additivity. The most notable missing piece is the complexity of $\{2,c\}$-submodular valuations with $c$ odd. %The integer-ratio case is tractable even for bi-valued submodular valuations, and our $\{-2,0,c\}$ hardness shows that negative marginals destroy tractability under submodularity.
Whether the half-integer case remains polynomial for positive submodular valuations is therefore open, and it is the case that would complete the picture. Further directions include closing the approximation gap for $p\ge3$ (between $pq/(pq-1)\le 12/11$ and our $\sqrt{3/2}$), extending the matching approach to leximin and other welfare objectives.

\paragraph{Use of generative AI.}
The authors used GenAI iteratively to refine and improve an initial draft of this paper. The authors have verified all proofs and accept full responsibility for the accuracy and correctness of the results presented in the paper.

\bibliographystyle{plain}
\bibliography{references}

\FloatBarrier
\appendix

\section*{Appendix}

\section{Missing proofs from Section~\ref{sec:bivalued}}

\begin{proof}[Proof of Lemma~\ref{lem:form}]

Define the \emph{alternating sequence}
\[
  \nu_j=\begin{cases} s & \text{if } T-q(j-1) \text{ is odd},\\ s-1 & \text{if } T-q(j-1) \text{ is even},\end{cases}\qquad \text{ for all }1\le j\le K.
\]
%Because $q$ is odd, consecutive values $T-q(j-1)$ have opposite parity, so $\nu$ alternates between $s$ and $s-1$, for instance, it starts with $\nu_1=s$ if and only if $T$ is odd.
We first prove that 
\begin{equation}\label{eq:nu_mu}
    \mu_j=\nu_j,~\text{ for }j\le K-1,~\text{ and } ~~0\le\mu_K\le\nu_K.
\end{equation}
Indeed, let $1\le j\le K-1$ and $x=T-q(j-1)$.  Since $j<K$ we have $qj<T$, so $x-q=T-qj>0$ and
\begin{align*}
  \mu_j= \lceil\frac{(T-q(j-1))^+}{2}\rceil -\lceil\frac{(T-qj)^+}{2}\rceil-1=\lceil \frac{x}{2}\rceil-\lceil\frac{x-q}{2}\rceil-1.
\end{align*}
If $x$ is even, $x-q$ is odd, so 
\[
\mu_j=\frac{x}{2}-\frac{x-q+1}{2}-1=\frac{q-1}{2}-1=s-1.
\]
If $x$ is odd, $x-q$ is even, so 
\[
\mu_j=\frac{x+1}{2}-\frac{x-q}{2}-1=\frac{q+1}{2}-1=s.
\]
In both cases $\mu_j=\nu_j$.

Let $y=T-q(K-1)$. By the choice of $K$, $1\le y\le q$. We have $g_T(K-1)=K-1+\lceil y/2\rceil$ and $g_T(K)=K$, so $\mu_K=\lceil y/2\rceil-1\ge0$. If $y$ is odd then $\nu_K=s$ and $\mu_K=\frac{y-1}{2}\le\frac{q-1}{2}=s$. If $y$ is even then $y\le q-1=2s$, $\nu_K=s-1$ and $\mu_K=\frac y2-1\le s-1$.
%\begin{claim}\label{lem:shape}
%\leavevmode
%For $j\le K-1$, $\mu_j=\nu_j$, and $0\le\mu_K\le\nu_K$.
%\end{claim}
%In words, the marginal savings follow the alternating template $\nu$ exactly up to the $(K-1)$-th high-valued good, and the $K$-th such good saves at least nothing and at most its template value. Beyond $K$, each additional high-valued good only increases the number of goods needed, since $K$ high-valued goods alone already reach $T$.

We now prove the claim in the lemma.

(i) If $T$ is odd, then $\mu_1=\nu_1=s$ by (\ref{eq:nu_mu}) and definition of $\nu_j$. Since a pair slot must begin with $s-1$, this term becomes a single slot of weight $s$. Similar, if $T$ is even, then $\mu_1=\nu_1=s-1$, so there is no lead single slot. 

(ii) Let $j_0\in\{1,2\}$ be the first remaining position. By (\ref{eq:nu_mu}), $\mu_{j_0},\mu_{j_0+1},\dots$ follows $\nu$, which starts with $s-1$ at $j_0$ and alternates. Group the positions $j_0,j_0+1,\dots,K$ into consecutive blocks $(s-1,s)$ of $\nu$, the last block possibly consisting of a single position.
 
Every block lying entirely within positions $1,\dots,K-1$ equals $(s-1,s)$ in $\mu$ and becomes a pair slot. Consider the block containing $K$:
\begin{itemize}%[nosep]
  \item If $K$ is the first position of its block, then $\nu_K=s-1$, and $\mu_K\le s-1$ by (\ref{eq:nu_mu}). It becomes one tail slot of weight at most $s-1$.
  \item If $K$ is the second position of its block, then $\mu_{K-1}=s-1$ and $\nu_K=s$. If $\mu_K=s$ the block becomes a pair slot. Otherwise $\mu_K\le s-1$, and the block becomes two tail slots of weights $s-1$ and $\mu_K\le s-1$, in this order.
\end{itemize}
This proves (ii). The last claim holds by construction: the parse-order filling with $k$ goods collects exactly $\mu_1,\dots,\mu_k$, because a pair slot holding one good contributes $s-1=\mu_j$ for its first position $j$, and holding two contributes $(s-1)+s$.
\end{proof}

\begin{proof}[Proof of Lemma~\ref{lem:slots}]
    If $K\le1$ there is at most one slot and the claim is immediate. Let $K\ge2$ and use the form of Lemma~\ref{lem:form}. Let $\lambda_{\max}\in\{0,1\}$ indicate whether a lead slot exists, and let $d_t=(s-1)-w_t\ge0$ be the \emph{deficit} of tail slot $t$ with weight $w_t$.
 
    Consider any filling with $k$ goods. Let $\lambda\in\{0,1\}$ indicate whether the lead slot is filled, let $c$ be the number of pair slots holding two goods, and let $R$ be the set of filled tail slots. Since every good contributes $s-1$ except that the lead good contributes $s$, a full pair earns an extra $1$, and a tail good contributes $s-1-d_t$, the value of the filling is
\begin{equation}\label{eq:value}
    (s-1)\,k\;+\;\lambda\;+\;c\;-\;{\sum}_{t\in R}d_t .
\end{equation}
    The {\em parse-order filling} with $k \le K$ goods places one good in the slot position corresponding to each of $\mu_1, \dots, \mu_k$, and leaves all remaining positions empty. Moreover, by the decomposition $\cS$, one can see that the parse-order filling with $k$ goods has value $\mu_1+\cdots+\mu_k=\sigma(k)$ for every $0\le k\le K$. Therefore, it remains to show that~\eqref{eq:value} is at most the value of the parse-order filling. We distinguish two cases.
 
    \emph{Case $k\le\lambda_{\max}+2P$.} The parse-order filling uses no tail slot, fills the lead slot if $\lambda_{\max}=1$ and $k\ge1$, and then completes $c^*=\lfloor(k-\lambda^*)/2\rfloor$ pairs, where $\lambda^*=\min(\lambda_{\max},k)$; note $c^*\le P$. Its value is $(s-1)k+\lambda^*+c^*$. For an arbitrary filling, $\lambda\le\lambda^*$ (since $\lambda\le\lambda_{\max}$ and $\lambda\le k$) and $2c\le k-\lambda$, so
    \[
    \lambda+c\;\le\;\lambda+\lfloor\tfrac{k-\lambda}{2}\rfloor\;\le\;\lambda^*+\lfloor\tfrac{k-\lambda^*}{2}\rfloor\;=\;\lambda^*+c^*,
    \]
    where the middle inequality holds because $\lambda\mapsto\lambda+\lfloor(k-\lambda)/2\rfloor$ is nondecreasing. As $\sum_{t\in R}d_t\ge0$, the claim follows.
 
    \emph{Case $k>\lambda_{\max}+2P$.} The lead and pair slots hold at most $\lambda_{\max}+2P$ goods, so every filling uses at least $\tau^*=k-\lambda_{\max}-2P$ tail slots, i.e.\ $|R|\ge\tau^*$. Hence
    \[
     \lambda+c-{\sum}_{t\in R}d_t\;\le\;\lambda_{\max}+P-\bigl(\text{sum of the }\tau^*\text{ smallest deficits}\bigr).
    \]
    The parse-order filling fills the lead slot, all $P$ pairs, and the first $\tau^*$ tail slots, which by Lemma~\ref{lem:form}(c) have the largest weights and hence the smallest deficits. It therefore attains this bound.~
\end{proof}

%\begin{proof}[Proof of Claim~\ref{claim:pq}]
%   We have $q(p-k)=p\ell$, so $p$ divides $q(p-k)$. Since $\gcd(p,q)=1$, $p$ divides $p-k\in\{0,\dots,p\}$, hence $k\in\{0,p\}$.
%\end{proof}

\section{Missing proofs from Section~\ref{sec:submod}}

\begin{proof}[Proof of Lemma~\ref{lem:matroid-val}]
The rank function is submodular and $|S|$ is modular. A rank
marginal is either zero or one, so a marginal of $v$ is respectively
$-2$ or $c$. Also $v(\varnothing)=0$, and the displayed identity
follows by rearranging terms.
\end{proof}

\begin{proof}[Proof of Lemma~\ref{lem:sparse-paving}]
Any two distinct members of $\mathcal F$ differ in at least one whole
pair, so they intersect in at most $r-2$ elements. Such a family of
$r$-subsets is the family of circuit-hyperplanes of a sparse paving
matroid, obtained from the uniform matroid $U_{r,|E|}$ by declaring
these sets dependent; see \citet{Oxl11}. We give the basis argument
to verify the construction directly.

Let $\mathcal B$ consist of all $r$-subsets of $E$ outside
$\mathcal F$. This family is nonempty: the $r$-subsets of any
$(r+1)$-set intersect pairwise in $r-1$ elements, so at most one
belongs to $\mathcal F$. To check basis exchange, take
$B_1,B_2\in\mathcal B$ and $x\in B_1\setminus B_2$. If
$|B_2\setminus B_1|=1$, exchanging $x$ for its unique element gives
$B_2$. Otherwise the sets
$(B_1\setminus\{x\})\cup\{y\}$ for $y\in B_2\setminus B_1$
intersect pairwise in $r-1$ elements. At most one is forbidden, so
at least one is in $\mathcal B$. Thus $\mathcal B$ is the basis
family of a rank-$r$ matroid.

Every $(r-1)$-set extends to a member of $\mathcal B$, since its
distinct $r$-element extensions intersect in $r-1$ elements and
there are at least two such extensions. Hence every smaller set is
independent. An $r$-set has rank $r-1$ exactly when it belongs to
$\mathcal F$. Finally, every set of more than $r$ elements contains
an $(r+1)$-set and therefore a member of $\mathcal B$. These facts
give the displayed rank function.
\end{proof}

\begin{proof}[Proof of Claim~\ref{clm:submod-size}]
	There are $k+1$ agents and $m=2k+kd+d_0+H
	=2k^2+(c+4)k-(c+2)s$ 
	items. The pair-agent valuations are additive. The collector's rank
	requires only set counting and one membership test in
	$\mathcal F_{\mathcal G}$, which is polynomial-time computable as explained
	above. The graph $\mathcal G$, the item classes, and these evaluation rules
	therefore provide a polynomial-size explicit representation.
\end{proof}

\section{Missing proofs from Section~\ref{sec:m101}}
\label{app:g-slots}
\begin{proof}[Proof of Lemma~\ref{lem:conflict-pair-submod}]
Fix a bundle $S$ and an item $o\notin S$, and let $\Delta(o\mid S)=v_i(S\cup\{o\})-v_i(S)$.
Adding $o$ changes the additive part by $v_i(o)$. In the second sum, only terms of pairs containing $o$ can change, and since the pairs in $\Pi_i$ are disjoint, $o$ lies in at most one pair.
If $o$ lies in no pair, then $\Delta(o\mid S)=v_i(o)$.
If $o=x$ for some $(x,y)\in\Pi_i$, the term $\mathbf 1[x\in S]\,\mathbf 1[y\in S]$ changes from $0$ to $\mathbf 1[y\in S]$, so $\Delta(x\mid S)=1-\mathbf 1[y\in S]$; symmetrically, if $o=y$, then $\Delta(y\mid S)=-\mathbf 1[x\in S]$.
In all cases
\begin{equation}\label{eq:cp-marginal}
\Delta(o\mid S) \;=\; v_i(o) - \mathbf 1\bigl[\,o \text{ is paired and } \bar o\in S\,\bigr].
\end{equation}
 
\emph{Submodularity.} Let $S\subseteq T$ and $o\notin T$. If $\bar o\in S$ then $\bar o\in T$, so the indicator in~\eqref{eq:cp-marginal} for $T$ is at least the one for $S$, and hence $\Delta(o\mid S)\geq\Delta(o\mid T)$.
 
\emph{Range of the marginal values.} By~\eqref{eq:cp-marginal}, $\Delta(o\mid S)=v_i(o)\in\{-1,0,1\}$ if $o$ is unpaired, $\Delta(x\mid S)=1-\mathbf 1[y\in S]\in\{0,1\}$, and $\Delta(y\mid S)=-\mathbf 1[x\in S]\in\{-1,0\}$.
\end{proof}

\begin{example}[A conflict-pair valuation]\label{ex:conflict-pair}
Consider an agent and five items $a,b,c,d,e$ with $v(a)=v(c)=1$, $v(b)=v(d)=0$, $v(e)=-1$, and a single conflict pair $\Pi=\{(a,b)\}$. Some values are presented in the following table.
\begin{center}
\small
\begin{tabular}{lcccccc}
\toprule
$S$ & $\{a\}$ & $\{b\}$ & $\{a,b\}$ & $\{c,d\}$ & $\{a,b,c,d\}$ & $\{a,b,e\}$ \\
\midrule
additive part  & $1$ & $0$ & $1$ & $1$ & $2$ & $0$ \\
\# conflicts & $0$ & $0$ & $1$ & $0$ & $1$ & $1$ \\
$v(S)$ & $1$ & $0$ & $0$ & $1$ & $1$ & $-1$ \\
\bottomrule
\end{tabular}
\end{center}
%If $(c,d)$ were also a conflict pair, then $v(\{a,b,c,d\}) = 2-2 = 0$.
 
The valuation $v$ is submodular: the marginal value of $b$ is $0$ at bundles without $a$ and $-1$ at bundles containing $a$, the marginal value of $a$ is $1$ at bundles without $b$ and $0$ at bundles containing $b$, and all other items contribute additively.
\end{example}

\begin{proof}[Proof of Lemma~\ref{lem:conflict-construction}]
For (a), fix an agent $a_{v,k}$. Set
\[
v_{a_{v,k}}(g_{v,k})=v_{a_{v,k}}(g_{v,k+1})=1,
\]
and set $v_{a_{v,k}}(y_C)=0$ if an occurrence in $C$ is assigned to this agent. Give all remaining items singleton value $-1$. Let $\Pi_{a_{v,k}}$ consist of the single prescribed pair, or be empty if no occurrence is assigned. The resulting valuation has exactly the form in Definition~\ref{def:conflict-pair}. In its conflict pair, the first item has singleton value $1$ and the second has singleton value $0$. As the agent has at most one pair, the required disjointness is automatic. Agents without an assigned occurrence have additive valuations. Even if a variable occurs twice in the same clause, the two occurrences are assigned to different agents. Lemma~\ref{lem:conflict-pair-submod} therefore applies to every agent.

For (b), let
\[
\ell=\sum_{v=1}^{N}o_v\leq 3M
\]
be the total number of literal occurrences. There are $n=\sum_v L_v$ agents, $n$ cycle goods, and $M$ clause items. Since $\max\{2,o_v\}\leq o_v+2$, we obtain
\[
n\leq\ell+2N\leq3M+2N,
\qquad m=n+M\leq4M+2N.
\]
If every variable occurs at least twice, then $n=\ell\leq3M$ and $m\leq4M$. Each valuation is specified by its two positive goods, at most one zero-valued clause item, and at most one conflict pair, with default singleton value $-1$. Thus the representation and its construction are polynomial in the source size; writing the full singleton-value matrix is also polynomial.
\end{proof}

%\begin{proof}[Proof of Claim~\ref{lem:shape}]
%(i) If $k\ge K$ then $qk\ge T$, so $\Phi_T(k)=k\ge K=\Phi_T(K)$.
%\end{proof}

\end{document}